\documentclass{amsproc}

\usepackage{algorithmic}
\usepackage{algorithm}
\usepackage{listings}
\usepackage{graphicx}
\usepackage{graphics}
\usepackage{url}
\usepackage{amsmath}
\usepackage{amsthm}
\usepackage{amssymb}
\usepackage{complexity}
\usepackage{bm}
\usepackage{units}
\usepackage{xcolor}
\usepackage{colortbl}
\usepackage{mathdots}
\usepackage{amsrefs}

\newtheorem{theorem}{Theorem}[section]

\newtheorem{corollary}[theorem]{Corollary}

\newtheorem{property}[theorem]{Property}

\theoremstyle{definition}
\newtheorem{definition}[theorem]{Definition}
\newtheorem{example}[theorem]{Example}

\theoremstyle{remark}
\newtheorem{remark}[theorem]{Remark}

\numberwithin{equation}{section}
\renewcommand{\thefigure}{\arabic{figure}}

\newcommand{\sgn}{\text{sgn}}

\numberwithin{equation}{section}
\numberwithin{equation}{subsection}

\begin{document}
\title[Combinatorial Spaces And Order Topologies]{Combinatorial Spaces And Order Topologies}
\author{Philon Nguyen}
\address{Technical Report}
\date{Montreal, 2012.} 

\begin{abstract}
An archetypal problem discussed in computer science is the problem of searching for a given number in a given set of numbers. Other than sequential search, the classic solution is to sort the list of numbers and then apply binary search. The binary search problem has a complexity of $O(logN)$ for a list of $N$ numbers while the sorting problem cannot be better than $O(N)$ on any sequential computer following the usual assumptions. Whenever the problem of deciding partial order can be done in $O(1)$, a variation of the problem on some bounded list of numbers is to apply binary search without resorting to sort. The overall complexity of the problem is then $O(log R)$ for some radius $R$. The following upper-bound for finite encodings is shown:
$$
O(\log\lvert\textbf{X}\rvert_\infty \log \log\ N)
$$
Also, the topology of orderings can provide efficient algorithms for search problems in combinatorial spaces. The main characteristic of those spaces is that they have typical space complexities of $O(2^N)$, $O(N!)$ and $O(N^N)$. The factorial case describes an order topology that can be illustrated using the combinatorial polytope . When a known order topology can be combined to a given formulation of a search problem, the resulting search problem has a polylogarithmic complexity. This logarithmic complexity can then become useful in combinatorial search by providing a logarithmic break-down. These algorithms can be termed as the class of search algorithms that do not require read and are equivalent to the class of logarithmically recursive functions. Also, the notion of order invariance is discussed.
\end{abstract}

\maketitle

\section{Computable Structures}

A few terminological remarks can be made on the algebra of computable structures.

\begin{remark}
A computable space is defined using the distinguishability problem.
\end{remark}

\begin{remark}
A computable space is said to be complete if it can represent a Turing-complete space.
\end{remark}

Furthermore, the following can be noted:

\begin{remark}
The equality operator is defined using the cut metric of Equation \ref{cutd}. This is denoted $\underset{\mathcal{N}}{=}$.
\end{remark}

Whenever $\mathcal{N}=0$, the conventional equality operator can be used.

\begin{remark}
Let $\textbf{T}^n$ be a sequence such that $|\textbf{T}|=n$. A computable sequence is such that the following holds:
\begin{equation}
f(\textbf{T}^n) = \textbf{S}
\end{equation}
for some computable function $f$ and for a set $\textbf{S}$ such that:
\begin{equation}
\textbf{T}^n_i\neq \textbf{T}_j^n,\ \ \textbf{S}_i\neq \textbf{S}_j
\end{equation}
and:
\begin{equation}
\textbf{T}^n_i=\textbf{T}_j^n,\ \ \textbf{S}_i=\textbf{S}_j
\end{equation}
\end{remark}

\begin{remark}
Let $f$ be a computable function, then $f$ has a representation in a complete computable space.
\end{remark}

\begin{definition}\label{fac}
A domain $\mathbb{F}^n$ on a function $f$ is termed a factorial domain if the following holds:
\begin{equation}
f(\textbf{x}_i) \underset{\mathcal{N}}{\neq} f(\textbf{x}_j),\ \ \textbf{x}_i \neq \textbf{x}_i
\end{equation}
and:
\begin{equation}
f(\textbf{x}_i) \underset{\mathcal{N}}{=}f(\textbf{x}_j),\ \ \textbf{x}_i = \textbf{x}_i
\end{equation}
\end{definition}

\section{Order Topologies}
An order can be best described as a subset of the integral line. Let $\textbf{X}$ define a discrete space such that $f:\textbf{X}\rightarrow\textbf{Y}$ defines an orderable discrete space $\textbf{Y}$. $\text{Or}(\textbf{Y})$ then define the order of $\textbf{Y}$ and a compact order is such that:
\begin{equation}\label{compact}
\frac{\Delta\text{Or}(\textbf{Y})}{\Delta f(\textbf{Y})} = 1
\end{equation}
Whenever $\textbf{Y}$ is termed orderable, then the following holds:
\begin{equation}
\textbf{Y}_i-\textbf{Y}_j=k
\end{equation}
for some value $k$. The difference operator $\Delta\textbf{Y}$ is then defined on a differentiable representation of $\textbf{X}$. Equation \ref{compact} defines the sequence of integers $(n,n+1,\ldots,n+k)$ or, alternatively, a sequence of ordered incremental values $(a_0, a_0+\delta,\ldots,a_0+n\delta)$ for some unit incremental operator $\delta$. This can be denoted $\text{Or}_1(\textbf{X})$. Given an arbitrary order $\text{Or}(\textbf{Y})$, the following holds:
\begin{equation}
\text{Or}_1(\textbf{Y})=\pi(\text{Or}(\textbf{Y}))
\end{equation}
for some permutation $\pi$. Whenever the permutation can be found \emph{seeminglessly}, it can be said that $\text{Or}_1(\textbf{Y})$ is equivalent to $\text{Or}(\textbf{Y})$. 
\begin{figure*}[t]
\centering
{\includegraphics[height=5cm]{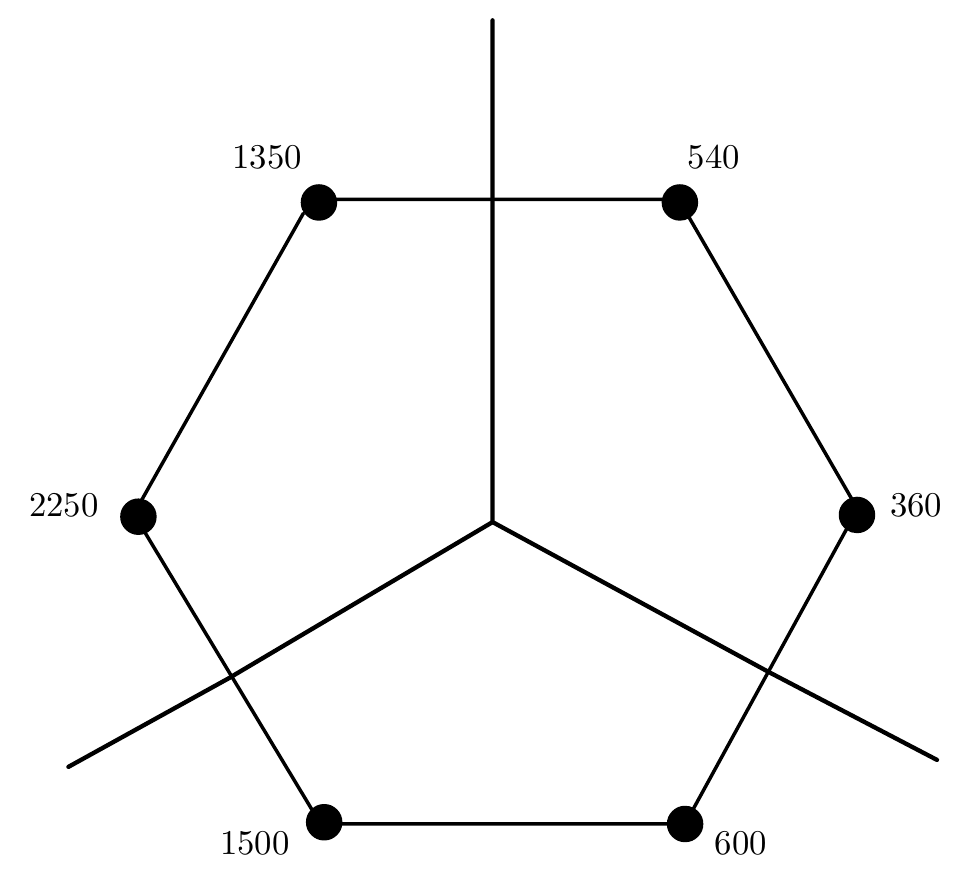}}
\caption{Polytope of $\text{Or}(\pi(1,2,3))$ in a prime factorial domain.}
\label{poly}
\end{figure*}
\begin{example}
The space of rank orders is given by the set of all permutations of the form $\pi(1,2,\ldots,n)$ \cite{schulman:geo}. Figure \ref{poly} illustrates the Cartesian representation of an order over a prime factorial domain given by:
\begin{equation}
\text{G}^3_3(\pi_i(a_1,a_2,a_3))=p_1^{a_{i,1}}p_2^{a_{i,2}}p_3^{a_{i,3}}
\end{equation} 
where the $\text{G}$-operator defines a $g$-code of length $3$ on a set of cardinality $3$. A candidate distance measure in Cartesian space which preserves $\text{Or}\circ\text{G}^3_3(\pi(1,2,3))$ is given by:
\begin{equation}
d(\textbf{x},\textbf{y})=\lVert \textbf{w}\textbf{x}-\textbf{w}\textbf{y} \rVert
\end{equation}
for some given weight vector $\textbf{w}$ and $\textbf{x},\textbf{y}\in\pi(1,2,3)$. It follows that:
\begin{equation}
\text{Or}_1(\pi(1,2,3))=\pi(\text{Or}\circ\text{G}^3_3(\pi({1,2,3})))
\end{equation}
for $\textbf{p}=(2,3,5)$. Equivalently, an $l$-code representation can be given as:
\begin{equation}
\begin{split}
\text{L}^3_3(\pi_i(a_1,a_2,a_3))=&\left\lfloor{a_{i,1}}\log p_1\right\rfloor_{\log}+\\
                                                  &\left\lfloor{a_{i,2}}\log p_2\right\rfloor_{\log}+\left\lfloor{a_{i,3}}\log p_3\right\rfloor_{\log}
\end{split}
\end{equation} 
where the log-floor operator keeps a logarithmic number of digits in regards to the largest value in the code. In practice, this can be computed using a digit-by-digit radix representation.
\end{example}

\begin{example}
Let the following log-prime factorial domain on an index function $\mathcal{I}=(1,\ldots,2^N)_2$ be defined as:
\begin{equation}\label{lcode}
\text{L}^n_N\left(\mathcal{I}_i\right)=\sum_{j=1}^{N}a_i \left\lfloor\log p_i\right\rfloor_{\log},\ \ a_i\in(0,1)
\end{equation}
It follows that the subsets of the $l$-codes of Equation \ref{lcode} are given by:
\begin{equation}
\begin{split}
\text{L}^n_N\left(\mathcal{I}\left(\textbf{1}_N\times\mathcal{I}_i\right)\right) &=\mathcal{I}_i\text{L}^n_N\left(\mathcal{I}\right)\\
                                                                                                                   &=\sum_{i=j}^{N}k_j \left\lfloor\log p_j\right\rfloor_{\log}
\end{split}
\end{equation}
for $ k_j\in(0, \ldots, N)$ and some arbitrary sequence of ordered primes. An illustration of the case $N=4$ is shown in Figure \ref{korder}.
\end{example}

Let $\textbf{X}^N$ define the product space $\textbf{X}\times\ldots\times\textbf{X}$ such that $\textbf{Y}^N\subseteq\textbf{X}^N$. If there exists some ordering differentiable to one, then then the order is termed cyclical. For some given $l$-code representation, we have:
\begin{equation}
\frac{\Delta\text{Or}(\pi_{i,j,k\ldots}(\text{L}^n_N(\textbf{X})))}{\Delta \text{L}(\textbf{X}}) = 1
\end{equation}
given some reshaping permutation operator $\pi_{i,j,k\ldots}$
\begin{example}
The indicator function $\mathcal{I}=(1,2,\ldots,2^N)_2$ has a $(g,l)$-code representation such that any point on in $\text{Or}(\text{L}^n_N(\mathcal{I}))$ can be computed using $O(N\log 2)$ recursions, or alternatively, $O(1)$ when the order is known.
\end{example}

\begin{example}
The $O(2KM(K+M))$ Jackson-Sheridan-Tseitin transform \cite{jackson:cnf,tseitin:cnf} is given by:
\begin{equation}
\textbf{S} = \left[
\renewcommand{\arraystretch}{2}
\begin{array}{c|c}
    \left(\textbf{I}_K^T\otimes\textbf{1}_2\right)^T & \textbf{B} \\ \hline
     \textbf{0} &\left(\textbf{I}_M^T\otimes\textbf{1}_2\right)^T 
  \end{array}\right]
\end{equation}
for some random matrix $\textbf{B}$ such that $b_{i,j}\in(0,1)$
\end{example}
\begin{figure*}[t]\label{korder}
\centering
{\includegraphics[height=6.5cm]{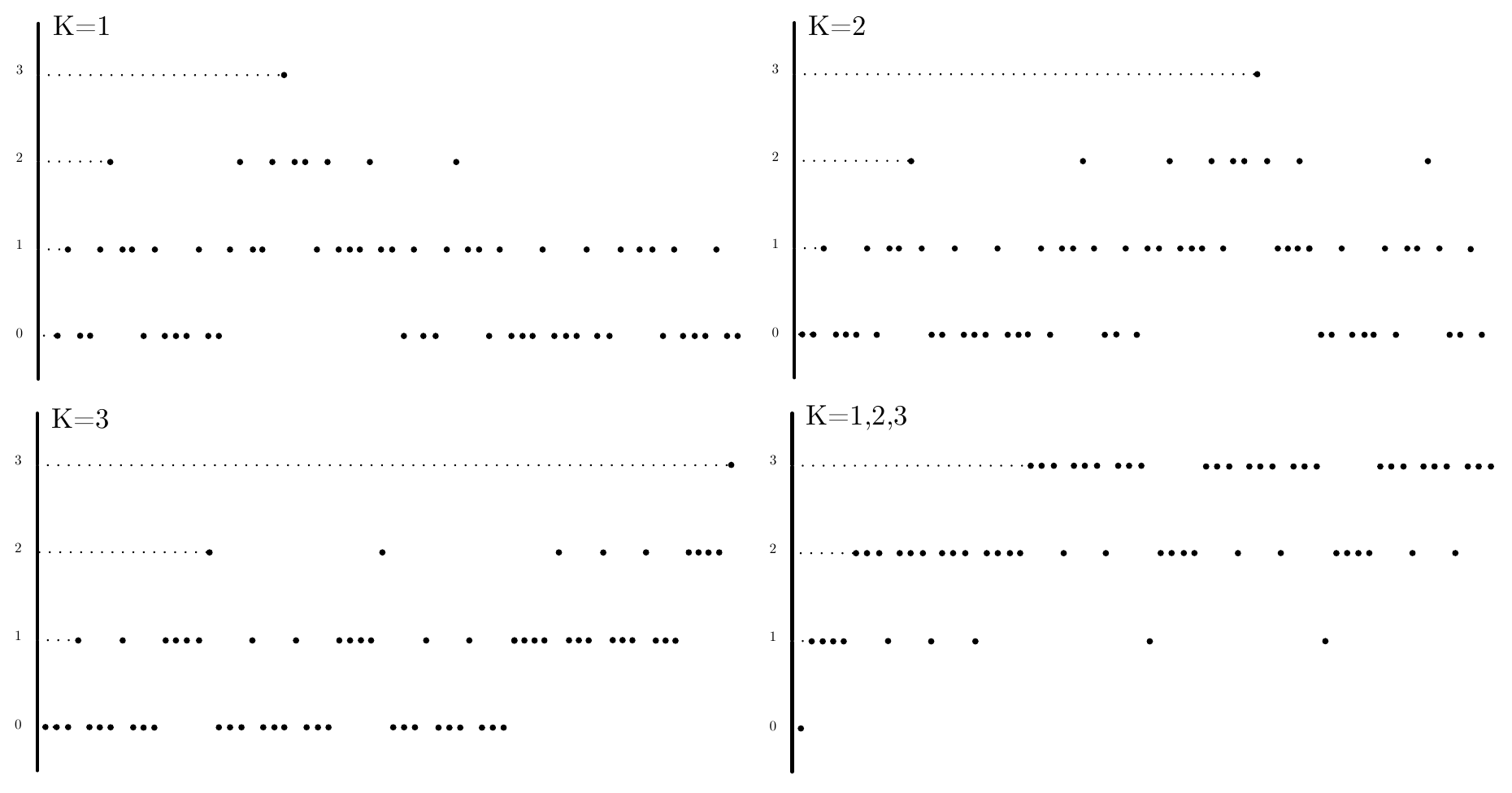}}
\caption{Order curve.}
\end{figure*}

\section{Distance Equations}

Formally, the notions of distance spaces, norms and metrics are defined in the following definitions.

\begin{definition}\label{distance}
A distance measure $d:\textbf{X}\times\textbf{X}\to \mathbb{A}$ is a non-invertible mapping that has the following properties: 	
\begin{equation}\label{zero}
d({{\mathbf{x}}_{i}},{{\mathbf{x}}_{i}})=0,\;{{\mathbf{x}}_{i}}\in \mathbf{X}
\end{equation}		
\begin{equation}\label{inverse}
d({{\mathbf{x}}_{i}},{{\mathbf{x}}_{j}})=d({{\mathbf{x}}_{j}},{{\mathbf{x}}_{i}}),\;{{\mathbf{x}}_{i}},{{\mathbf{x}}_{j}}\in \mathbf{X}		
\end{equation}	
for some non-empty set \textbf{X} and an alphabet $\mathbb{A}$ isomorphic to an arbitrary subset of $\mathbb{R}_{+}^{{}}$, (X,d) defines a distance space.
\end{definition}

Given a distance matrix $\mathbf{D}\in \mathbb{R}_{+}^{N\times N}$ on some finite discrete space \textbf{X}, the mapping  $f \circ g$ is invertible whenever $g:\mathbf{X}\times \mathbf{X}\to \Delta $ is invertible. The distance measure on the constructed space is then given by $f$.

\begin{theorem}
A distance space is uniquely defined by $\Delta $. 
\end{theorem}

\begin{figure*}[t]
\centering
\includegraphics[height=4cm]{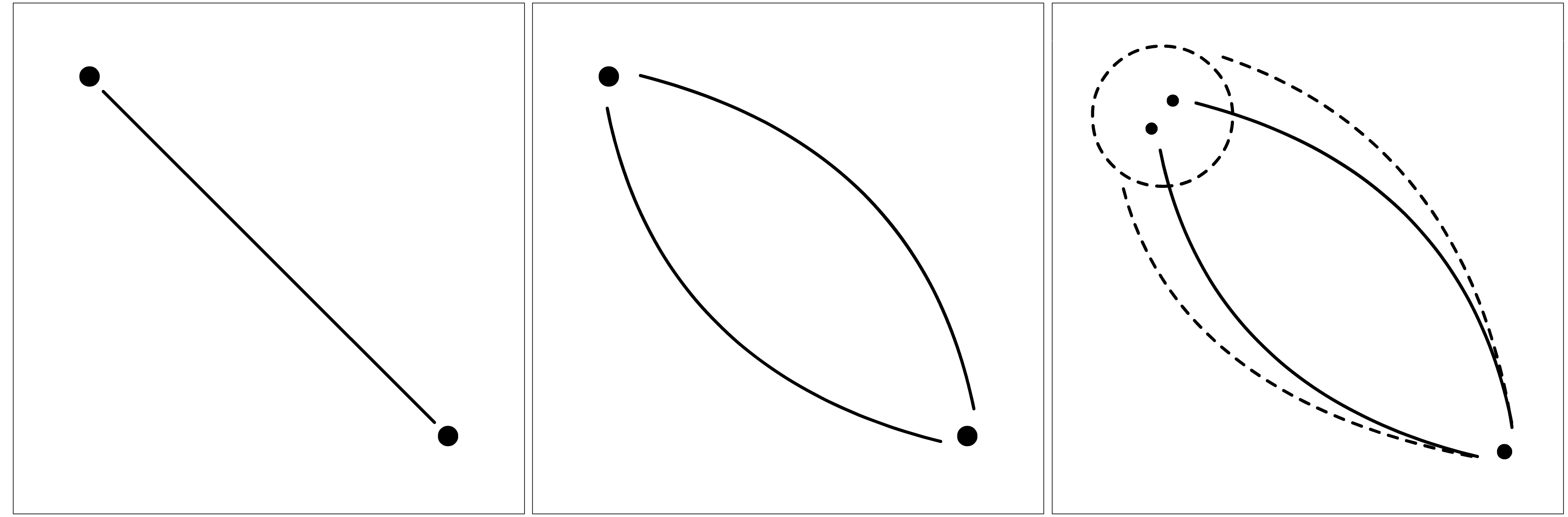}
\caption{Relation diagrams.}
\end{figure*}

By convention, $(\mathbf{X},d)$ or $\Delta $ then denote an arbitrary distance space. It can also be shown that $\Delta \in {{\mathbb{R}}}$, $\Delta \in {{\mathbb{R}}^{2}}$, $\Delta \in \mathbb{R}_{+}^{N\times N+1}$ and $\Delta \in \mathbb{R}\cup \mathbb{R}_{+}^{N\times N}$ are sufficient, under different assumptions. Any binary function that satisfies Equation \ref{inverse} can form a distance space by imposing the following:
\begin{equation}
d(\mathbf{x}_i,\mathbf{x}_j) = f'\left(f(\mathbf{x}_i,\mathbf{x}_j),f(\mathbf{x}_j,\mathbf{x}_i)\right)
\end{equation}	
for some function $f’$ that satisfies Equation \ref{inverse}. Other such measures have been used. A distance space embeds a collection of elements onto the real line and induces a poset topology and an order topology.

\begin{definition} 
Given some non-empty set $\textbf{X}$, a metric space $(\textbf{X}, d)$ is a space equipped with a metric distance measure d defined using Equations \ref{zero} and \ref{inverse} with the additional property:
\begin{equation}
d(\mathbf{x}_i,\mathbf{x}_k) \le d(\mathbf{x}_i,\mathbf{x}_j)+d(\mathbf{x}_j,\mathbf{x}_k),\  \mathbf{x}_i,\mathbf{x}_j,\mathbf{x}_k\in\mathbf{X}	
\end{equation}	 
\end{definition}

Metricity preserves heuristic notions such as distance minimization. Similarly, semimetrics and ultrametrics can be defined by constraining or varying the standard properties. 

\begin{definition}
A norm $\lVert\cdot\rVert$ is a function $f:\mathbf{X}\rightarrow\mathbb{A}$ such that:
\begin{equation}
\big\lVert k\mathbf{x}_i\big\rVert = \big\lvert k \big\rvert \big\lVert \mathbf{x}_i \big\rVert
\end{equation}
\begin{equation}
\big\lVert \mathbf{x}_i+\mathbf{x}_j\big\rVert \leq \big\lVert \mathbf{x}_i \big\rVert +  \big\lVert \mathbf{x}_j \big\rVert  
\end{equation}
\begin{equation}
\big\lVert \mathbf{0}\big\rVert = \mathbb{A}_0
\end{equation}
for some zero vector $\mathbf{0}$, a zero element $\mathbb{A}_0$, $k\in\mathbf{K}$ and $\mathbf{x}_i,\mathbf{x}_j\in\mathbf{X}$.  $\mathbf{X}$ then defines a normed space. 
\end{definition}

Norms preserve metricity and linearity. Also, the notion of mapping of metric spaces is defined as follows.

\begin{definition}
A mapping $f:\textbf{X}\rightarrow\textbf{Y}$ of a space $(\textbf{X},d)$ into a space $(\textbf{Y},d')$ exists if the following holds:
\begin{equation}	                                                 
 \frac{1}{{{k}_{1}}}d({{\mathbf{x}}_{i}},{{\mathbf{x}}_{j}})\le {d}'(f({{\mathbf{x}}_{i}}),f({{\mathbf{x}}_{j}}))\le {{k}_{2}}d({{\mathbf{x}}_{i}},{{\mathbf{x}}_{j}})
\end{equation}
such that ${{\mathbf{x}}_{i}},{{\mathbf{x}}_{j}}\in \mathbf{X}$, for some ${{k}_{1}},{{k}_{2}}\ge 1$. 
\end{definition}

An isomorphic mapping is then defined as a (1,1)-distortion mapping. An isometric mapping then defines an isomorphic mapping on a metric space. Mappings of metric spaces can be generalized to mappings of arbitrary spaces using the same definition.

On finite discrete spaces, mappings are given by the set of functions $f:\mathbb{R}_{+}^{N\times N}\to \mathbb{R}_{+}^{N\times N}$ such that some parameters $(k_1, k_2)$ hold. 

Whenever $(k_1, k_2)$-distortions can be defined, the distance matrix of the image is said to be bounded by the distance matrix of the domain on some arbitrary function. In the context of similarity search, distortions usually refer to a lower-bounding distortion whenever ${{k}_{2}}<1$. An upper-bounding distortion is then given by ${{k}_{2}}\ge 1$. 

In the context of compression, the usual metric is given on some vector space by:
\begin{equation}
k=\lVert f(\textbf{x})-\textbf{x}\rVert_{\infty },\;x\in \textbf{X} 
\end{equation}
The distortion is then either $(k,k)$, $({{k}_{1}},k)$ or $(k,{{k}_{2}})$. 

Typical application domains of mappings include computational geometry, approximation algorithms and functional analysis. 

Other problems of practical importance include mappings of random walks into the plane, mappings into binary codes and mappings onto the $n$-sphere.

A finite discrete space is described in Definition \ref{deffds} and refers to finite computable discrete spaces. 

\begin{definition}\label{deffds}
Let $f:\mathbf{X}\to \mathbf{Y}$ be a mapping from some arbitrary space $(\mathbf{X},d)$ into a space $(\mathbf{Y},d)$.
 Whenever the following holds for all ${{\mathbf{y}}_{i}},{{\mathbf{y}}_{j}}\in \mathbf{Y}$:
\begin{equation}
\sum\limits_{i,j}{d}({{\mathbf{y}}_{i}},{{\mathbf{y}}_{j}})=0
\end{equation}
given a cut metric defined as:
\begin{equation}\label{cutd}
d(\textbf{y}_i,\textbf{y}_j) = \left\{ \begin{array}{rl}
				         1, &  \mathcal{N}(\textbf{y}_i)\cap\mathcal{N}(\textbf{y}_j)\neq\O
\\
				         0, & \mathcal{N}(\textbf{y}_i)\cap\mathcal{N}(\textbf{y}_j)=\O
				       \end{array}\right.
\end{equation}
and some arbitrary neighborhood $\mathcal{N}$, $(\mathbf{Y},d)$ defines a discrete space. When $\mathbf{Y}$ is finite, $f$ generates a finite discrete space. $(\mathbf{X},d)$ is then defined as the represented space.
\end{definition}

\begin{corollary} 
A finite discrete space can be constructed if and only if its elements have an arbitrary neighborhood that can be computed exactly in the classical sense.
\end{corollary}

\begin{corollary} 
A countable subset of the real line containing arbitrary transcendentals is not a computable discrete space whenever the value of the transcendentals are considered. 
\end{corollary}

\begin{proof}
Let $\mathbf{X}=\{{{\mathbf{x}}_{1}},\ldots ,{{\mathbf{x}}_{n}}\}$ define some space $(\mathbf{X},d)$ for some transcendentals ${{\mathbf{x}}_{i}}$. Given that the halting problem and the machine equivalence problem are undecidable, the following problem cannot be decided for arbitrary values:
\begin{equation}
\mathcal{N}({{\mathbf{x}}_{i}})\cap \mathcal{N}({{\mathbf{x}}_{j}})=\varnothing 
\end{equation}
\end{proof}

Equivalently, Equation \ref{cutd} cannot be decided on a classical machine for arbitrary subsets of the real line.

\begin{corollary}
The problem of defining any discrete space can be decided upon whenever an arbitrary relation on all its elements can be decided. An arbitrary graph is then defined.
\end{corollary}

\begin{example}
The space defined by:
\begin{equation}
\mathbf{X}=\bigcup\limits_{i=1}^{N}{\left\{ \pi +1/i,1/i \right\}}
\end{equation}
is a computable finite discrete space under the following distance matrix:
\begin{equation}
\textbf{D} = \left[
\renewcommand{\arraystretch}{1.2}
\begin{array}{c|cccc}
      & i=1 & \cdots & 2N & \cdots \\ \hline
   i=1  & 0& \cdots & a & \cdots\\
   i+1  & b & \cdots &\lvert a - b \rvert & \cdots\\
    \vdots & \vdots & \ddots & \vdots & \cdots\\
    2N &  a& \cdots &0 &\cdots \\
    2N+1 & \lvert a+b  \rvert &  \cdots&  b  & \cdots \\
      \vdots & \vdots & \vdots  & \vdots  & \vdots  \\
  \end{array}\right]
\end{equation}
for some element $a$ of an arbitrary alphabet and b given as follows:
\begin{equation}
b=\frac{1}{i}-\frac{1}{1+i}
\end{equation}
such that $a\pm b\in {{\mathbb{R}}_{+}}$.
\end{example}

\begin{example}
 Let the following arbitrary curve be defined on the following sequence of iterates:
\begin{equation}
{{f}_{i}}(x)={{f}_{i-1}}(x)+\delta 
\end{equation}
for some infinitessimal operator $\delta $. Then, $f$ defines a discrete space under the following distance matrix:
\begin{equation}
\left[ 
\renewcommand{\arraystretch}{1.2}
\begin{array}{c|cc}
   {} & {{f}_{0}}(x) & \ldots   \\\hline
   {{f}_{0}}(x) & 0 & \cdots   \\
   \vdots  & \vdots  & \cdots   \\
   {{f}_{i}}(x) & \sum{{{\delta }_{i}}} & \cdots   \\
   \vdots  & \vdots  & \ddots   \\
\end{array} \right]	
\end{equation}
\end{example}

\begin{example}
A discrete space can be defined on the following sequence of iterates:
\begin{equation}
{{f}_{i}}(x)={{f}_{i}}(x)+\Omega\delta 
\end{equation}
for some uncomputable number $\Omega \in [0,1]$ \cite{calude:omega}. A distance matrix can then be given as: 
\begin{equation}
\left[ 
\renewcommand{\arraystretch}{1.2}
\begin{array}{c|cc}
   {} & {{f}_{0}}(x) & \ldots   \\\hline
   {{f}_{0}}(x) & 0 & \cdots   \\
   \vdots  & \vdots  & \cdots   \\
   {{f}_{i}}(x) & \Omega \sum{{{\delta }_{i}}} & \cdots   \\
   \vdots  & \vdots  & \ddots   \\
\end{array} \right] 	
\end{equation}
\end{example}

\section{Distance Equations On Finite Discrete Spaces}

Some known notions often associated to search algorithms can be formalized using the distance equation concept. 

An example is similarity search in some $n$-dimensional space equipped with the ${{\ell }_{p}}$ metric, which has found many interesting algorithmic solutions.

\begin{definition}
Let $\mathbf{X}\cup \mathbf{q}$ be some finite discrete space. A distance equation is defined as follows:
\begin{equation}
\delta (\mathbf{q}):_{{}}^{{}}\;d(\mathbf{X},\mathbf{q})+k=0
\end{equation}
for all values of $k\in \mathbb{K}$ such that $\mathbb{K}\subseteq \mathbb{R}$.
\end{definition}

\begin{corollary}
Any distance equation on some finite discrete space defines a set of open or closed intervals.
\end{corollary}

\begin{corollary}
A distance equation is defined on a continuous space if and only if $k\in \mathbb{K}$ for some $\mathbb{K}\subseteq\bar{\mathbb{R}}$.
\end{corollary}

The complexity of defining a distance equation on a discrete space can then be seen as equivalent to the complexity of defining an equivalent distance equation on a continuous space whenever the corresponding representations of $\mathbf{X}$ are of the same order.

A distance equation is then said to define a $k$-partition of the real numbers.

\begin{definition}
Two distance equations are isomorphic if they define two $k$-partitions $p,q$ such that $\lvert p\rvert = \lvert q\rvert$.
\end{definition}

\begin{definition}
A composite distance equation is defined using standard set operations on distance equations:
\begin{equation}
\left( \bigcup ,\bigcap  \right)_{{}}^{{}}\delta (\mathbf{q}):\;_{{}}^{{}}d(\mathbf{X},\mathbf{q})+k=0
\end{equation}
\end{definition}

The solution of the equation is the subset of X that satisfies the equation and defines subregions of the distance space induced by the order topology. 

Whenever the set is viewed as a matrix, it is given by:
\begin{equation}
{{\mathbf{X}}_{\mathcal{I}}}=\mathcal{I}\mathbf{X}		
\end{equation}
for some indicator function $\mathcal{I}$. And more specifically, by absorption:
\begin{equation}
{{\mathbf{X}}_{\mathcal{I}}}={{\mathcal{I}}_{\mathbf{i}}}{{\mathbf{X}}_{\mathbf{i}}},\;\;\;\mathbf{i}\subset \mathbb{N}		
\end{equation}

\begin{theorem}
Let $\mathbf{X}$ be a finite discrete space, for any distance equation $\delta $ defined on a computable distance function, $\mathcal{I}$ is well-defined.
\end{theorem}

Also, if $\delta $ falls in the complexity class $C$ whenever ${{\mathcal{I}}_{i}}$ can be computed in $O(1)$, then it is customary to denote the resulting complexity class as ${{C}^{\delta }}$. 

By extension, any low-complexity computation of ${{\mathcal{I}}_{i}}$ denotes a complexity class ${{C}^{\delta }}$.

Let ${{\delta }_{0}}(\mathbf{q})$ define a distance equation with a solution equal to an arbitrary element chosen uniformly at random from the subset of $\mathbf{X}$ satisfying the adjoint $\delta $ equation. This is also referred to as the decision version of the search problem. 

The complexity of solving a distance equation $\delta $ is equal to a multiple of the complexity of solving its corresponding ${{\delta }_{0}}$ equation. The multiple is in the range $[0,|\mathbf{X}|]$. 

\subsection{Examples}

The similarity search problem in some finite discrete space $\mathbf{X}$ equipped with the ${{\ell }_{1}}$ metric is usually defined as follows:
 \begin{equation}\label{sim1}
\delta (\mathbf{q}):\;_{{}}^{{}}\|\mathbf{X}-\mathbf{q}{{\|}_{1}}\le\epsilon 
\end{equation}
for some $\epsilon \in {{\mathbb{R}}_{+}}$. The same problem using ${{\ell }_{2}}$ metrics is given as:
\begin{equation}\label{sim2}
\delta (\mathbf{q}):_{{}}^{{}}\;\|\mathbf{X}-\mathbf{q}{{\|}_{2}}\le\epsilon
\end{equation}	
Similarly, the orthogonal range search problem can be defined as follows:
\begin{equation}\label{sim3}
\bigcup\limits_{i}{^{{}}\delta }({{\mathbf{q}}_{i}}):_{{}}^{{}}\;\|\mathbf{X}-{{\mathbf{q}}_{i}}{{\|}_{\infty }}{}_{{}}^{{}}\le _{{}}^{{}}\epsilon 
\end{equation}	
The random convex polytope search problem, also called the point location in arrangement of hyperplanes problem, is done on the discrete space $\mathbf{X}$ defined by the following arrangement:
\begin{equation}\label{simpoly}
\bigcup\limits_{i}{^{{}}\mathcal{A}}:\;_{{}}^{{}}{{\mathbf{A}}_{i}}\mathbf{q}\le{{\mathbf{b}}_{i}}
\end{equation}	
for some non-trivial matrices ${{\mathbf{A}}_{i}}$ with elements $a_{j,k}^{i}$ and $\mathbf{q},\mathbf{b}\in {{\mathbb{R}}^{n}}$.  Equation 2.1.4 can be translated to a distance equation such as: 
\begin{equation}\label{polydim}
\bigcup\limits_{i}{^{{}}\delta }(\mathbf{q}):_{{}}^{{}}\;\dim\mathbf{q}+\sum\limits_{j}{\text{ }}{{\sgn }^{*}}(\mathbf{A}_{j}^{i}{{q}_{j}}-b_{j}^{i})=0
\end{equation}
It can also be noted that the ${{\sgn }^{*}}$ function is here defined on the vertices of the $n$-cube ${{\mathbb{H}}^{n}}$ as:
\begin{equation} 
\sgn^*= \left\{ \begin{array}{rl}
				          1, & x>0  \\
  					   -1, & x\le 0  
				       \end{array}\right. 
\end{equation}
The solution of Equation \ref{polydim} returns the indexes $i$ of the discrete space \textbf{X} that solves the distance equation, or alternately the polytopes defined by $({{\mathbf{A}}_{i}},{{\mathbf{b}}_{i}})$. 

It has been shown to be equivalent to the search formulation of a Jackson-Sheridan-Tseitin transform. 

Let $(\mathbf{Y},d)$ be some finite discrete space such that ${{\mathbf{X}}_{i}},\mathbf{Q}\subseteq \mathbf{Y}$. Then, using the Hausdorff distance, the following distance equation can be defined:
\begin{equation}\label{simh}
\bigcup\limits_{i}{^{{}}\delta }(\mathbf{Q}):_{{}}^{{}}\;\|{{\mathbf{X}}_{i}}-\mathbf{Q}{{\|}_{H}}{{+}^{{}}}k=0
\end{equation}
for $k\in [{{\epsilon }_{1}},{{\epsilon }_{2}}]$. The Hausdorff distance is defined as usual:	
\begin{equation}
\|\mathbf{X}-\mathbf{Y}{{\|}_{H}}=\max \left\{ {{\sup }_{\mathbf{x}\in \mathbf{X}}}{{\inf }_{\mathbf{y}\in \mathbf{Y}}}d(\mathbf{x},\mathbf{y}),\;{{\sup }_{\mathbf{y}\in \mathbf{Y}}}{{\inf }_{\mathbf{x}\in \mathbf{X}}}d(\mathbf{x},\mathbf{y}) \right\}
\end{equation}
The overall complexity of Equations \ref{sim1}, \ref{sim2} and \ref{sim3} can be done, loosely, in at most $O(|\mathbf{X}|\dim\mathbf{X})$, where $|\mathbf{X}|$ denotes the cardinality of the set enclosing the space, following the usual assumptions. 

The overall complexity of Equation \ref{simh} is at most $O(|\mathbf{X}||\mathbf{Q}|\dim\mathbf{X})$. On the other hand, Equation \ref{simpoly} is said to have a search complexity of $O({{n}^{5}}\log N)$ \cite{goodman:handbook, salton:vector}.

\begin{example}
Let the finite discrete space $\textbf{X}\subset\mathbb{R}^3$ be an arbitrary set of $N$ points defining some arbitrary surface using some spherical coordinates $r$, $\theta$, $\phi$. The following radial basis function interpolates the nonlinear transformation $f:\textbf{X}\rightarrow\mathbb{S}^2$ on an orthogonal grid $r\times\theta\times\phi$:
\begin{equation}
f(\textbf{x}) = \sum_{k=1}^N c_k \varphi\left(\left\lvert \frac{1}{\lVert\textbf{x}\rVert} - \frac{1}{\lVert\textbf{x}_{k}\rVert} \right\rvert\right) 
\end{equation} 
where the interpolating coefficients are given by:
\begin{equation}
\textbf{c}=\Phi^{-1}\textbf{X}
\end{equation}
for some radial basis matrix $\Phi(\textbf{x})$ defined as:
\begin{equation}
\begin{bmatrix}
\varphi\left(\left\lvert \frac{1}{\lVert\textbf{x}_1\rVert} - \frac{1}{\lVert\textbf{x}_{1}\rVert} \right\rvert\right)  &  \varphi\left(\left\lvert \frac{1}{\lVert\textbf{x}_1\rVert} - \frac{1}{\lVert\textbf{x}_{2}\rVert} \right\rvert\right)  & \ldots & \varphi\left(\left\lvert \frac{1}{\lVert\textbf{x}_1\rVert} - \frac{1}{\lVert\textbf{x}_{N}\rVert} \right\rvert\right)  \vspace{2 mm} \\ 
\varphi\left(\left\lvert \frac{1}{\lVert\textbf{x}_2\rVert} - \frac{1}{\lVert\textbf{x}_{1}\rVert} \right\rvert\right)  & \varphi\left(\left\lvert \frac{1}{\lVert\textbf{x}_2\rVert} - \frac{1}{\lVert\textbf{x}_{2}\rVert} \right\rvert\right) & \ldots & \varphi\left(\left\lvert \frac{1}{\lVert\textbf{x}_2\rVert} - \frac{1}{\lVert\textbf{x}_{N}\rVert} \right\rvert\right)  \\
\vdots & \vdots & \ddots & \vdots \\
\varphi\left(\left\lvert \frac{1}{\lVert\textbf{x}_N\rVert} - \frac{1}{\lVert\textbf{x}_{1}\rVert} \right\rvert\right)  & \varphi\left(\left\lvert \frac{1}{\lVert\textbf{x}_N\rVert} - \frac{1}{\lVert\textbf{x}_{2}\rVert} \right\rvert\right)  & \ldots & \varphi\left(\left\lvert \frac{1}{\lVert\textbf{x}_N\rVert} - \frac{1}{\lVert\textbf{x}_{N}\rVert} \right\rvert\right)  \\
\end{bmatrix}
\end{equation}
The following distance equation defined on the finite discrete spaces $\textbf{X}\subseteq\textbf{X}'$ then returns the subregion corresponding to the interpolation of $\textbf{X}$:
\begin{equation}
\delta(\textbf{q}=\textbf{0}):\ \big\lVert f(\textbf{X}')-\textbf{q}\big\rVert = 1
\end{equation}
All points that can be interpolated to $\textbf{X}$ using the radial basis function approximation are mapped to the boundary of the unit circle centered at the origin.
\end{example}
\begin{example}
Sibuya's method for generating uniform random points on $\mathbb{S}^{n-1}$, the unit $n$-sphere, is given as follows:
\begin{equation}
f(x_i) = \left\{ \begin{array}{rl}
				        \cos{(2x_i\pi)}(x_i-x_{i-1})^{0.5} , & i\ is\ odd\\
 					\sin{(2x_i\pi)}(x_i-x_{i-1})^{0.5}, & i\ is\ even
				       \end{array}\right.
\end{equation}
where $x_i$ is a uniform random value in $[0,1]$ \cite{sibuya:sphere}. 

The distance distributions of the Sibuya distribution are well-defined and follow the results of the concentration of measure phenomenon. As dimension tends to infinity, the distance distribution tends towards a constant. The distribution then defines a countable subset of $\mathbb{S}^{n-1}$. 

As dimension tends to infinity, an isometric mapping of the$n$-dimensional Sibuya distribution in $(\mathbb{S}^1\cup\textnormal{\textbf{0}}, \ell_p)$, the set of points on the 2-dimensional centered circle equipped with the $\ell_p$ norm can be given as:
\begin{equation}
d(\textbf{x},\textbf{y}) = k\big\lVert \textbf{x}-\textbf{0} \big\rVert_p
\end{equation}
where $\textbf{0}$ denotes the zero vector , $k$ is a given constant and $\textbf{x}$ a vertex in the polygon projection of the simplex mapping of the Sibuya distribution. 

The numbers generated by the distribution can then be seen as a cyclical permuation.
\end{example}
\subsection{Classification, Equivalences and Dilations}
\subsubsection{Notation}

For some distance measure d, a search problem denoted $\delta $ on some finite discrete space $\mathbf{X}$ that can be solved using $d$ is said to be in the class $\mathcal{C}(\delta ,\mathbf{X},d)$. 
Also, let the dilation of an arbitrary mapping $f$ of $(\mathbf{X},d)$ into $(\mathbf{Y},{d}')$ be given on the following relation:
\begin{equation}
\frac{1}{c_1} \max_i{\big\lvert  \textbf{X}_i \big\rvert}  \leq  \max_i{\big\lvert  \textbf{Y}_i \big\rvert} \leq c_2 \max_i{\big\lvert  \textbf{X}_i \big\rvert}
\end{equation}
for some ${{c}_{1}},{{c}_{2}}\ge 1$ and for all ${{\mathbf{q}}_{i}}$ and for all subsets ${{\mathbf{X}}_{i}}\subseteq \mathbf{X}$ and ${{\mathbf{Y}}_{i}}\subseteq \mathbf{Y}$ such that:
\begin{equation}\label{c1}
\delta(\textbf{q}_i):\ d(\textbf{X}_i,\textbf{q}_i) + k = 0
\end{equation}
and:
\begin{equation}\label{c2}
{\delta }'({{\mathbf{q}}_{i}}):{{\;}^{{}}}{d}'({{\mathbf{Y}}_{i}},f({{\mathbf{q}}_{i}}))+{k}'=0
\end{equation}
and:
\begin{equation}\label{c3}
f({{\mathbf{X}}_{i}})\subseteq {{\mathbf{Y}}_{i}}
\end{equation}
If there is a $(1,O(1))$-dilation, then $\mathcal{C}(\delta ,\mathbf{X},d)$ is said to be equivalent to $\mathcal{C}({\delta }',\mathbf{Y},{d}')$. The later is denoted as follows:
\begin{equation}
\mathcal{C}(\delta ,\mathbf{X},d)\Rightarrow \mathcal{C}({\delta }',\mathbf{Y},{d}')
\end{equation}
More specifically, ${{c}_{2}}$ is expected to grow as follows:
\begin{equation}
{{c}_{2}}=1+\frac{\log |\mathbf{X}|}{|\mathbf{X}|}
\end{equation}
$\mathcal{C}(\delta ,\mathbf{X},d)$ is said to be probabilistically equivalent and bounded to $\mathcal{C}({\delta }',\mathbf{Y},{d}')$ whenever dilations can be probabilistically bounded. This is denoted:
\begin{equation}
\mathcal{C}(\delta ,\mathbf{X},d)\underset{k}{\mathop{\to }}\,\mathcal{C}(\delta ,\mathbf{Y},{d}')
\end{equation}
Two finite discrete distance spaces $(\mathbf{X},d)$ and $(\mathbf{Y},{d}')$ such that $f:\mathbf{X}\to \mathbf{Y}$ are isomorphically embeddable if the following holds:
\begin{equation}
\big\lvert  \textbf{Y}_i \big\rvert = \big\lvert  \textbf{X}_i \big\rvert
\end{equation}
for all ${{\mathbf{q}}_{i}}$ and for all subsets ${{\mathbf{X}}_{i}}\subseteq \mathbf{X}$ and ${{\mathbf{Y}}_{i}}\subseteq \mathbf{Y}$ such that Equations \ref{c1}, \ref{c2} and \ref{c3} hold. Given two distance measures $d$ and $d’$ on some finite discrete spaces $\mathbf{X}$ and $\mathbf{Y}$, if $(\mathbf{X},d)$ embeds isomorphically into $(\mathbf{Y},{d}')$, then $\mathcal{C}(\delta ,\mathbf{X},d)$ is said to be strictly equivalent to $\mathcal{C}({\delta }',\mathbf{Y},{d}')$. This is denoted as follows:
\begin{equation}
\mathcal{C}(\delta ,\mathbf{X},d)\Leftrightarrow \mathcal{C}({\delta }',\mathbf{Y},{d}')
\end{equation}		

\subsubsection{Examples}

Let the following equation be the general formulation of similarity search on some finite discrete metric space $(\mathbf{X},d)$ \cite{ciaccia:mtree,faloutsos:gemini,goodman:handbook}:
\begin{equation}\label{simsimple}
\delta (\mathbf{q}){{:}^{{}}}\;d(\mathbf{X},\mathbf{q})\le \epsilon 	
\end{equation}
with $\mathbf{X}\subset {{\mathbb{R}}^{n}}$. Given some $f:\mathbf{X}\to \mathbf{Y}$, it follows that there exists finite discrete spaces $\mathbf{Y}$ for which Equation \ref{simsimple} is equivalent to is \ref{sim1}, \ref{sim2} and \ref{sim3}. The set of finite discrete spaces for which this is the case is such that the following holds:
\begin{equation}
\begin{split}
O(\log |\mathbf{X}|) =\  &{{\max }_{i}}\left| \mathcal{B}_{\epsilon }^{n}\left( f({{\mathbf{q}}_{i}}) \right)\cap {{\mathbf{Y}}^{{}}} \right|\text{ }{{-}^{{}}}\\
&{{\max }_{i}}\left| \mathcal{B}_{\epsilon {{k}_{1}}}^{n}({{\mathbf{q}}_{i}})\cap {{\mathbf{X}}^{{}}} \right|
\end{split}
\end{equation}
for all possible values of $\mathbf{q}\in \mathbf{X}$, some ${{\ell }_{p}}$ ranges $\mathcal{B}_{\epsilon }^{n}$ of dimension $n$ and radius $\epsilon $ with distortion value $k$. Also, the following holds whenever a $({{k}_{1}},{{k}_{2}})$-distortion is defined:
\begin{equation}\label{balls}
\begin{split}
{{\max }_{i}}\left| \mathcal{B}_{\epsilon }^{n}\left( f({{\mathbf{q}}_{i}}) \right)\cap {{\mathbf{Y}}^{{}}} \right|\le\ & {{\max }_{i}}\left| \mathcal{B}_{\epsilon {{k}_{1}}}^{n}({{\mathbf{q}}_{i}})\cap {{\mathbf{X}}^{{}}} \right|-\\
&{{\max }_{i}}\left| \mathcal{B}_{\epsilon /{{k}_{2}}}^{n}({{\mathbf{q}}_{i}})\cap {{\mathbf{X}}^{{}}} \right|\text{ }
\end{split}
\end{equation}

Let $f:\mathbb{R}_{*}^{n}\to {{\mathbb{H}}^{n}}$ be a mapping to the vertices of the n-cube defining the following nearest neighbor problem:
\begin{equation}
f(\mathbf{q})=\{\mathbf{v}\in {{\mathbb{H}}^{n}}:\;{{\min }_{i}}\|{{q}_{i}}-{{v}_{i}}\|\}
\end{equation}	
The mapping $f$ chooses the vertex of ${{\mathbb{H}}^{n}}$ that is closest in terms of component-wise distances. The point location in arrangement of hyperplanes problem described by Equation \ref{polydim} can then be rewritten as follows:
\begin{equation}
\delta (\mathbf{q}){{:}^{{}}}\;\dim\mathbf{q}+\dim\mathbf{q}\cos \left( f(\mathbf{q}),\mathbf{1} \right)=0
\end{equation}
Therefore, Equation \ref{polydim} is strictly yet perhaps trivially equivalent to a linear function of the cosine distance, given some prior nearest-neighbor mapping. Equivalence to ${{\ell }_{p}}$ can be shown if the following holds:
\begin{equation}
\max \|\mathbf{X}-\mathbf{q}{{\|}^{{}}}\le k
\end{equation}	
for some constant $k$. The cosine distance is then strictly equivalent to a bounded ${{\ell }_{p}}$ metric, denoted ${{\bar{\ell }}_{p}}$, which in turn defines a subset of ${{\mathbb{R}}^{n}}$ of finite radius. It is also metrizable by using a convergent sequence that will yield an integral over a domain given by the range of the distance space. The range is then unbounded and cosine distances can be said to be strictly equivalent to ${{\ell }_{p}}$ metrics for finite discrete spaces. One such mapping for finite vector spaces that can serve this purpose is the limit:
\begin{equation}\label{dist2}
\cos(\textbf{x},\textbf{y})=\int_{0}^{\tan^2\nicefrac{\theta}{2}} \frac{\lVert\textbf{x}-\textbf{y}\rVert_p^{1/2}}{1+\lVert\textbf{x}-\textbf{y}\rVert_p}\, d\lVert\textbf{x}-\textbf{y}\rVert_p
\end{equation}
whenever Equation \ref{dist2} generates a computable distance space. As the differential element tends to infinity, the integral yields the range of the cosine. The formulas provided by S. Ramanujan for the computation of $\pi$ can also serve as an example to the metrization of directional distance spaces.

\subsection{Classic Theorems On Discrete Spaces}

The class of distance equations that are equivalent or strictly equivalent are tightly coupled with the properties of their adjoint discrete spaces. Known results can be shown in finite vector spaces from the literature on the subject \cite{goodman:handbook}.

\begin{theorem}\label{tnorm}
In a finite vector space, all norms are equivalent given some mapping $f:\mathbf{X}\to \mathbf{Y}$.
\end{theorem}

Theorem \ref{tnorm} is denoted:
\begin{equation}
\mathcal{C}(\delta ,\mathbf{X},\|\cdot {{\|}_{i}})\Leftrightarrow \mathcal{C}(\delta ,\mathbf{Y},\|\cdot {{\|}_{j}})
\end{equation}

\begin{corollary}
Given any distance measures $d$ and d', there exist finite discrete spaces $\mathbf{X}$ and $\mathbf{Y}$ such that: 
\begin{equation}
\mathcal{C}(\delta ,\mathbf{X},d)\Leftrightarrow \mathcal{C}(\delta ,\mathbf{Y},{d}')
\end{equation}
\end{corollary}

\begin{proof}
This can be trivially shown by simple linear manipulations of $d$ and $d'$ on discrete spaces of cardinality two.
\end{proof}

\begin{theorem} \cite{goodman:handbook}
Let the cut metric be defined on some subset $\mathbf{X}\subseteq \mathbf{Y}$ as follows:
\begin{equation}
d(\textbf{X},\textbf{q}) =  
 \left\{ \begin{array}{rl}
				         1, & f(\textbf{X},\textbf{q})=k\\
 					0, & otherwise
				       \end{array}\right.
\end{equation}
Then, some distance metric d is such that:
\begin{equation}
\mathcal{C}(\delta ,\mathbf{Y},d)\Leftrightarrow \mathcal{C}(\delta ,\mathbf{Y},\|\cdot {{\|}_{1}})	
\end{equation}
if and only if d can be expressed as:
\begin{equation}
d\equiv \sum\limits_{\mathbf{X}\subseteq \mathbf{Y}}{{{c}_{i}}}d(\mathbf{X},\mathbf{q})
\end{equation}
for some nonnegative and non-zero coefficients ${{c}_{i}}$.
\end{theorem}

\begin{theorem}
Given an arbitrary distance measure $d$, a metric d' and some mapping $f:\mathbf{X}\to \mathbf{Y}$, the following holds for some finite vector spaces $\mathbf{X}\subset {{\mathbb{R}}^{n}}$ and $\mathbf{Y}\subset {{\mathbb{R}}^{m}}$ with $m=(n,3)$:  
\begin{equation}
\mathcal{C}(\delta ,\mathbf{X},d)\Leftrightarrow \mathcal{C}(\delta ,\mathbf{Y},{d}')	
\end{equation}
 \end{theorem}

\begin{corollary}
Given an arbitrary distance measure $d$ and a metric d' some finite vector spaces $\mathbf{X},\mathbf{Y}\subset {{\mathbb{R}}^{n}}$ and some mapping $f:\mathbf{X}\to \mathbf{Y}$ the following holds:
\begin{equation}
\mathcal{C}(\delta ,\mathbf{X},d)\Leftrightarrow \mathcal{C}({\delta }',\mathbf{Y},{d}')
\end{equation}
 \end{corollary}

\begin{theorem}\label{jl}
Any finite metric space $(\mathbf{X},d)$ is probabilistically equivalent and bounded to $(\mathbf{Y},\|\cdot {{\|}_{p}})$. This is denoted:
\begin{equation}
\mathcal{C}(\delta ,\mathbf{X},d)\underset{k}{\mathop{\to }}\,\mathcal{C}(\delta ,\mathbf{Y},\|\cdot {{\|}_{p}})
\end{equation} 	
for some $\mathbf{X}\subset {{\mathbb{R}}^{n}}$ and $\mathbf{Y}\subset {{\mathbb{R}}^{O(\log |\mathbf{X}|)}}$. 
\end{theorem}

Theorem \ref{jl} is such that the relationship between distortions and dilations for the similarity search equation can be given by Equation \ref{balls}. Such distortions include low-dimensional subset mappings and random projections.

\subsection{Classic Theorems On Subspaces and Subsets}

Known theorems on the representation of distance equations over vector spaces as distance equations over subspaces and subsets of vector spaces can be defined from the literature on the subject. 

\begin{theorem} \cite{faloutsos:gemini}
Let $(\mathbf{X},d)$ be a finite discrete metric space such that $\mathbf{X}\subset {{\mathbb{K}}^{n}}$ and let $\mathbf{Y}$ be a linear subspace of $\mathbf{X}$ defined as follows:
\begin{equation}
\|{{\mathbf{y}}_{i}}{{\|}_{0}}=k,\;\;\;{{\mathbf{y}}_{i}}\in \mathbf{Y}	
\end{equation}
and:
\begin{equation}
{{\mathbf{y}}_{i}}=\mathcal{I}{{\mathbf{x}}_{i}},\;\;\;{{\mathbf{x}}_{i}},{{\mathbf{y}}_{i}}\in \mathbf{X},\mathbf{Y}
\end{equation}	
for some constant $k\le n$ and an indicator function $\mathcal{I}$. Then the following holds:
\begin{equation}
d({{\mathbf{x}}_{1}},{{\mathbf{x}}_{2}})\ge d({{\mathbf{y}}_{1}},{{\mathbf{y}}_{2}}),\;\;\;({{\mathbf{x}}_{1}},{{\mathbf{x}}_{2}}),({{\mathbf{y}}_{1}},{{\mathbf{y}}_{2}})\in \mathbf{X},\mathbf{Y}
\end{equation} 	
\end{theorem}

\begin{corollary}\cite{loh:dual_full}
Whenever the indicator functions ${{\mathcal{I}}_{i}}$ describes a contiguous subset of components, the following holds:
\begin{equation}
\delta {{(\mathbf{q})}^{{}}}{{\subseteq }^{{}}}\bigcup\limits_{i}{^{{}}{\delta }'}({{\mathcal{I}}_{i}}\mathbf{q})) 
\end{equation}	
given $\delta $ and ${\delta }'$ defined as follows:
\begin{equation}
\delta (\mathbf{q}):\;\|\mathbf{X}-\mathbf{q}{{\|}_{p}}{}^{{}}{{\le }^{{}}}\epsilon 
\end{equation}
\begin{equation}
\bigcup\limits_{i}{^{{}}{\delta }'}(\mathcal{I}{{\mathbf{q}}_{i,j}}):{{\;}^{{}}}\|\mathcal{I}\mathbf{X}-{{\mathcal{I}}_{i}}\mathbf{q}{{\|}^{{}}}\le O({{n}^{\alpha }}\epsilon )
\end{equation}	
for $O(n/k)$ indicator functions and a constant $\alpha >0$.
\end{corollary}

\begin{theorem}\cite{yu:high}
Let $\mathbf{X}\cup \mathbf{q}\subset {{\mathbb{K}}^{n}}$ be some finite discrete space, there exist $f,g:{{\mathbb{R}}^{n}}\to \mathbb{R}$ such that the following holds:
\begin{equation}
\bigcup\limits_{i}{^{{}}\delta }({{\mathbf{q}}_{i}}){}^{{}}{{\subseteq }_{{}}}\bigcup\limits_{i,j}{^{{}}{\delta }'}({{f}_{j}}({{\mathbf{q}}_{i}}))
\end{equation}
given $\delta $ and ${\delta }'$ defined as follows:
\begin{equation}
\bigcup\limits_{i}{^{{}}\delta }({{\mathbf{q}}_{i}}):{{\;}^{{}}}\|\mathbf{X}-{{\mathbf{q}}_{i}}{{\|}_{\infty }}{}^{{}}{{\le }^{{}}}\epsilon 
\end{equation}	
\begin{equation}
\bigcup\limits_{i,j}{^{{}}{\delta }'}({{f}_{j}}({{\mathbf{q}}_{i,j}})){{:}^{{}}}\;\left| g(\mathbf{X})-{{f}_{j}}({{\mathbf{q}}_{i,j}})\text{ } \right|=0
\end{equation}
\end{theorem}

\begin{theorem}
Let $f:\mathbf{X}\to \mathbf{Y}$ be a distortionless mapping given by a binary matrix $\mathbf{G}$ as follows:
\begin{equation}
f(\mathbf{X})=\mathbf{GX}
\end{equation} 	
for some $\mathbf{x}\in \mathbf{X}$. Then there exists $\mathbf{G}$ such that following holds for ${{\mathbf{y}}_{i}},{{\mathbf{y}}_{j}}\in f(\mathbf{X})$ and a constant $\epsilon $:
\begin{equation}
\|{{\mathbf{y}}_{i}}-{{\mathbf{y}}_{j}}{{\|}_{1}}{}^{{}}{{\ge }^{{}}}\epsilon ,\;\;\;{{\mathbf{y}}_{i}}\ne {{\mathbf{y}}_{j}}
\end{equation}     
$\mathbf{X}$ and $\mathbf{Y}$ can be chosen such that $\mathbf{X}\subseteq \mathbb{B}_{1}^{n},\mathbf{Y}\subseteq \mathbb{B}_{1}^{m}$.
\end{theorem}

\subsection{Characteristic Radius and Encodings}
As a converse to the notion of error-correcting bounds, a characteristic radius can be defined on any finite discrete distance space.
\begin{definition}
Given an isomorphic mapping $f:\mathbf{X}\to \mathbf{Y}$, the actual characteristic radius of $f$ is the minimal radius of $\mathbf{Y}$ in the distance space $(\mathbf{Y},d)$. 
\end{definition}
This is given by:
\begin{equation}
r=\min d\left( f({{\mathbf{x}}_{i}}),f({{\mathbf{x}}_{j}}) \right)
\end{equation}
for all ${{\mathbf{x}}_{i}},{{\mathbf{x}}_{j}}\in \mathbf{X}$.

Given a normed finite discrete space $\textbf{X}$, the compactness $c$ is then defined as:
\begin{equation}
 c(\textbf{X})=\lVert \textbf{X} \rVert_\infty - r
\end{equation}

The set of all isomorphic mappings of some finite discrete space is equivalent to the set of all injective mappings $g:f\to \mathbb{R}_{+}^{N\times N}$ that are order invariant given some distance function $f:\mathbf{X}\times \mathbf{X}\to \mathbb{R}_{+}^{N\times N}$. For some continuous function $f(x,y,z)$, this can be seen as the smoothing process given by:        
\begin{equation}
f(z)=g\circ f(z)
\end{equation}
for some value $z$ defining an arbitrary cutting plane. Or alternatively:
\begin{equation}
{{\left. {{f}^{-1}}-g\circ {{f}^{-1}} \right|}_{f(x,y)}}{{=}^{{}}}0
\end{equation} 	
which equivalently defines a projection to identity of $f$ and $g$ into ${{\mathbb{R}}^{2}}$. 

\begin{theorem}
Given some discrete bounded space, a computable characteristic radius can be defined.
\end{theorem}

\begin{corollary}
An unbounded discrete space has an actual characteristic radius of $\infty $.
\end{corollary}

\begin{corollary}
Any unbounded discrete space $(\mathbf{X},d)$ can be bounded through some mapping $f:(\mathbf{X},d)\to (\mathbf{Y},{d}')$.
\end{corollary}

Also, most practical results apply to some general finite discrete space or some non-trivial class of such spaces.
As the characteristic radius decreases, the complexity of the encoding is said to increase. The complexity is then given by function of the compactness $c(\mathbf{X})$.

For some finite discrete space in $\mathbb{R}$, a few additional results can be obtained.

\begin{theorem}
Any finite discrete space $\mathbf{X}\subset \mathbb{R}$ has a characteristic radius upper bounded by:
\begin{equation}
 \log \|\mathbf{X}{{\|}_{\infty }}
\end{equation}
\end{theorem}

\begin{corollary}\label{main}
For some finite value $n$, any code $\mathbf{X}$ defined in ${{\mathbb{N}}^{n}}$ has a minimal descriptive length upper bounded by $O(\log \|\mathbf{X}{{\|}_{\infty }})$.
\end{corollary}

These results will be further discussed and shown in subsequent sections. Also, Corollary \ref{main} can be seen as an information theoretical or algorithmic complexity bound on normed finite discrete spaces. Such bounds are directly proportional to the characteristic radius. A characteristic radius then qualifies how discrete a space is. It defines a compressibility coefficient and by extension a computability order.

\subsection{Distinguishability On The Line}
\begin{figure*}[t]
\centering
\includegraphics[height=5cm]{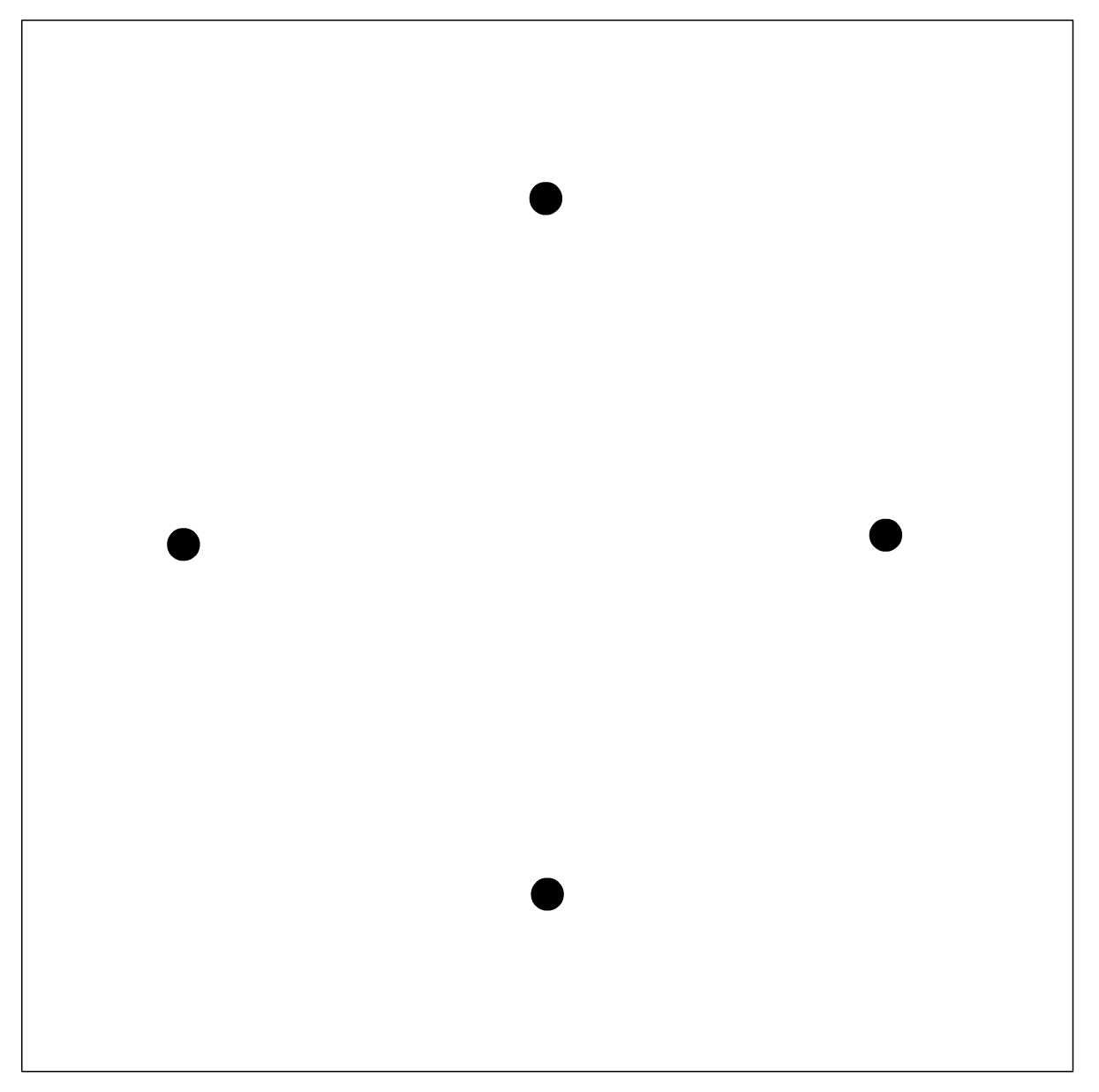}
\caption{Distinguishability on the line.}
\label{Dist}
\end{figure*}
Let $f$ define some function from some finite discrete space to a finitely computable subset of the real line. The distinguishability problem is then given equivalently to the derivation problem on the real line:
\begin{equation}\label{disting}
{{f}_{x}}=k 	
\end{equation}
whenever Equation \label{disting} generates a finite computable set.

It can be shown that the set of functions $f$ that satisfies the distinguishability problem are given by: 
\begin{equation}
{{f}_{x}}-f_{x}^{-1}=k
\end{equation}
whenever the equation is finite.

Given such a function $f$, reducing the complexity of the distinguishability problem can be done by finding the set of order preserving mappings of $f$ given by:
\begin{equation}
f{{\left. ^{-1}-g\circ {{f}^{-1}} \right|}_{f(x)}}{{=}^{{}}}0
\end{equation}
whenever $f$ is computable. 

\begin{definition}
Let $f$ be a norm such that \ref{disting} holds, then $f$ can be termed a distinguishing norm. 
\end{definition}

The definition of discreteness in terms of Equation \ref{disting} can then be restated as follows.

\begin{theorem}
Let $f:\mathbf{X}\to \mathbf{Y}$ be an asymptotically finite computable function, then $\mathbf{X}$ is a discrete space.
\end{theorem}

The complexity of $f$ defines the computable radius of the discrete space which is conceptually equivalent to the characteristic radius.

\section{Polylogarithmic Reduction In Finite Space}

Let the g-number $\text{G}_{N}^{n}$ of a sequence of $N$ integers be defined on some factorial domain.

\begin{definition} 
A $g$-code denotes and by extension defines the set of reversible functions over some factorial domain.
\end{definition}

A $g$-code uniquely defines an arbitrary sequence and is equivalent to the set of discrete spaces that can be encoded on a classical computer. Other than its celebrated use in logic and theoretical computer science, the uniqueness property of $g$-codes has been exploited in hashing algorithms. From the perspective of information and quantity, g-numbers typically have large logarithmic norms.

\begin{example}
The problem of finding the minimal number of bits required to encode a given sequence of elements from some finite discrete space using $g$-codes is directly related to known information theoretical bounds. Also, polylogarithmic bounds can be specified. This is shown for the particular cases of $\text{G}_{N}^{1}$, $\text{G}_{N}^{8}$, ${{\text{G}}_{N}}^{16,20}$, ${{\text{G}}^{32}}_{N}$ and ${{\text{G}}^{64}}_{N}$ and their associated distance spaces. 
\end{example}

\begin{example}
The general problem of determining the compressibility of a sequence has been shown to be uncomputable. More specifically, the algorithmic complexity of a sequence can be reduced to the Halting problem. From an information theoretic point of view, some results show that random sequences cannot improve on the logarithmic scale. Provided specific encodings and distributions, optimal bounds can be effectively be computed. By extension, while exact lower-bounds are in general uncomputable, known upper-bounds have been provided in the literature on the subject.
\end{example}

\begin{example}
Polylogarithmic upper-bounds can be provided to the compressibility problem using integral prime $g$-codes. The encoding and decoding of $g$-codes requires the rapid resolution of the constrained prime factorization problem on reasonably integers while the problem of factoring a large integer on a logarithmic scale can be related to the problem of deciding efficiently the integral indeterminate polynomial equation given by:
\begin{equation}\label{dio}
x-{{2}^{a}}=0
\end{equation}
for some $x\in \mathbb{N}$ and $a\in {{\mathbb{R}}_{+}}$. Whenever $a$ is a rational, the problem can be solved trivially. The decision problem of Equation \ref{dio} is then to determine if an arbitrary power of two accepts an integer solution and corresponds to the integrals solutions of the base two function on a real domain.

Also, it is known that the general solvability of integer-valued indeterminate polynomial equations is an undecidable problem. Also, in its simplest expression given by Equation \ref{dio}, the problem is \emph{very} decidable.
\end{example}

An $l$-code is typically defined as follows:

\begin{definition}
The l-number of $\text{L}_{N}^{n}$ of a sequence of cardinality $N$ is usually defined as:
\begin{equation}
\text{L}_{N}^{{n}}=f(\text{G}_{N}^{n})	
\end{equation}
for some funtion $f$.
\end{definition}

In general, an $g$-code on some finite discrete space $\mathbf{X}$ can be defined as:
\begin{equation}\label{lcode2}
\text{L}_{N}^{{n}}=f(\text{G}_{N}^{n})
\end{equation}
such that $r(\text{L}_{N}^{{n}})=O\left( \log \|\text{G}_{N}^{n}\| \right)$ for some given norm.

\begin{theorem}
$\text{L}_{N}^{{n}}$ defines a finite discrete space equipped with a characteristic radius.
\end{theorem}

\subsection{L-Codes}
From the perspective of compression, common results for $g$-codes can be discussed for ASCII, UTF, 32-bit encodings and 64-bit encodings.

\begin{theorem}
Any prime bitwise $g$-code has a decoding complexity of $O(N)$ using trial algorithms.
\end{theorem}

A bitwise $g$-code can then be denoted $\text{G}_{N}^{2}$. In general, the following complexity result can be stated.

\begin{theorem}
Given some basis $f(N)\in {{\mathbb{K}}^{N}}$, a $g$-code $\text{G}_{N}^{n}$ has a size complexity upper-bounded by $O(\log \|f(N)\|\|n{{\|}_{\infty }})$.
\end{theorem}

An $l$-code on ASCII codes can be denoted as $\text{L}_{N}^{8}$. The equivalent $l$-codes for UTF, 32-bit encodings and 64-bit encodings are respectively $\text{L}_{N}^{(16,20)}$, $\text{L}_{N}^{32}$ and $\text{L}_{N}^{64}$.

\begin{example}
Let an $l$-code $\text{L}_{N}^{n}$ of length $N$ on an n-dimensional integer code be defined as in Equation \ref{lcode2} for $f$ defining the logarithmic and inverse functions, a typical ordering of the $l$-code is then given by, for $N=kn$:

\begin{equation}
\begin{split}
({{a}_{1}},\ldots ,{{a}_{n}}),(({{a}_{1}},\ldots ,{{a}_{n}})+1)\bmod N,\ldots ,\\
(({{a}_{1}},\ldots ,{{a}_{n}})+N/n-1)\bmod N
\end{split}
\end{equation}
\end{example}

\begin{example}\label{lat}
The ordering of the linear $l$-codes $\text{L}_{N}^{16}=\sum{{{a}_{i}}\log {{p}_{i}}}$ and $\text{L}_{N}^{16}=\sum{a_{i}^{-1}\log {{p}_{i}}}$ are given by:
\begin{equation}      
\begin{split}                             
 (b_{1},\ldots ,b_{16})=&(1,17,33,49,65,81,97,113,129,145, \\
&161,177,193,209,225,241) 
\end{split}
 \end{equation}
\begin{equation}  
\begin{split}                               
(b_{1},\ldots ,b_{16})=&(241,225,209,193,177,161,145,129, \\
&113,97,81,65,49,33,17,1) 
\end{split}
\end{equation}

An order curve is then defined as the ordering of some factorial domain function on an indicator function. The order curve of $\text{L}^n_N(\mathcal{I})$ is shown in Figure \ref{orderc}.
\end{example}	

Example \ref{lat} defines an oriented lattice. Similar orders can be defined in $\|L_{N}^{n}-0{{\|}_{\infty }}$. The lattices of Example \ref{lat} have orientation 1  and -1 respectively.

\subsection{Jackson-Sheridan-Tseitin Transforms}
Let $\mathbf{S}$ be a finite discrete space defining a Jackson-Sheridan-Tseitin transform \cite{jackson:cnf,tseitin:cnf} such that ${{s}_{1}}<\ldots <{{s}_{n}}$ and:
\begin{equation}
\mathbf{S}'=\mathcal{I}\left( \mathbf{S}\otimes {{\mathbf{1}}_{{{2}^{n}}}} \right)	
\end{equation}
It can be shown that the order set of $\text{L}_{N}^{n}(\mathcal{I})$  generates  a  quasi-lattice  for $N={{2}^{n}}$. Let the quasi-lattice $\text{L}_{N}^{n}(\mathcal{I})$ be denoted $f(x)$ and let the following order curve be defined as $\Delta {{[f]}_{i}}(x)$ for some constant $h=k$. Then the following can be shown:
\begin{equation}
\Delta [f]_{i}(x)={{2}^{i}}+{{2}^{i+1}}(i-1)	
\end{equation}
Or alternatively:
\begin{equation}
i\equiv k\bmod {{2}^{k}}
\end{equation} 	

Let $\text{S}_{N}^{n}$ denote the order curve on $\text{L}_{N}^{n}(\mathcal{I})$ such that:
\begin{equation}
\text{S}_{N}^{n}=\text{L}_{N}^{n}(\mathbf{X})=\text{L}_{N}^{n}\left( \mathcal{I}\left( \mathbf{A}\otimes {{\mathbf{1}}_{{{2}^{N}}}} \right) \right)
\end{equation}
$\text{S}_{N}^{n}$ can be defined as follows:
\begin{equation}\label{func}
\begin{split}
\text{S}_{N}^{n} &= \sum_{i=1}^{N}\left\lfloor a_i f(s_i)\log p_i \right\rfloor_{\log}\\
 &= \sum_{i=1}^{N}\left\lfloor a_i \log p_i  \right\rfloor_{\log}+\Delta\text{S}^n_N(\mathcal{I})
\end{split}
\end{equation}
for some function $f({{s}_{i}})$, ${{a}_{i}}\in \{0,1\}$, ${{s}_{i}}\in \mathbf{S}$ and a logarithmic code floor function. $p_i$ then defines an arbitrary set of ordered primes.

\begin{remark}\label{constant}
A set of good constants are such that $f(s_i)\approx1$. In general, a set of \emph{good} constants are such that $f(s_{i,j,\ldots})\approx k_{i,j,\ldots}$ such that $\lvert \textbf{k}\rvert$ is bound (Equation \ref{dist}).
\end{remark}

Following this remark, whenever $f(s_i)\approx1$, the order of $\Delta \text{S}_{N}^{n}$ is equal to the order of $\text{L}_{N}^{n}$ for some l-function such that:
\begin{equation}
r\left( \Delta \text{S}_{N}^{n} \right)=o\left( \|\text{L}_{N}^{n}(\Delta \text{S}_{N}^{n}){{\|}_{\infty }} \right) 
\end{equation}	
This is denoted: 
\begin{equation}
\text{Or}(\Delta \text{S}_{N}^{n})=\text{Or}(\text{L}_{N}^{n})
\end{equation}
\begin{theorem}\label{oracle}
An ordered bounded space of cardinality $2^n$ can be searched in $O(n \log 2)$ whenever partial order can be decided in $O(1)$.
\end{theorem}
\begin{theorem}\label{permord}
Let $\text{Or}(\text{L}_{N}^{n})=\text{Or}(\text{S}_{N}^{n})$ on some function $f$ defined as in Equation \ref{func}, and let the $l$-code $\text{S}_{N}^{n}$ be defined such that the following holds:
\begin{equation}
\text{Or}(\text{L}_{N}^{n}(\pi {{(0,1,2,3,\ldots {{2}^{n}}-1)}_{2}}))=\text{Or}(\text{S}_{N}^{n}) 		
\end{equation}
A set of constants $f({{c}_{i}}{{s}_{i}})={{c}_{i}}f({{s}_{i}})$ can be found with space complexity ${{c}_{i}}=O(\log \|\mathbf{S}{{\|}_{\infty }}+n\log k)$.
\end{theorem}
\begin{corollary}
Let $\mathbf{X}=\{{{\mathbf{x}}_{1}},\ldots ,{{\mathbf{x}}_{N}}\}$ be some finite discrete space with order $\text{Or}(\mathbf{X})$. An order preserving function is such that for ${{\mathbf{F}}_{i,j}}=f({{\mathbf{x}}_{i}})$, there exists $f:\mathbf{F}\to {{\mathbf{K}}^{N}}$ such that $\text{Or}(f(\mathbf{X}))=\text{Or}(\text{L}^n_N(\mathcal{I}))$.	
\end{corollary}
\begin{remark}
Let $(\mathbf{S}/T,1)$ be an arbitrary Jackson-Sheridan-Tseitin transform with solutions ${{\mathcal{I}}_{s}}$. A general form for distance equations that retrieves a set ${{(\text{L}_{N}^{n})}_{s}}$ is:
\begin{equation}
\bigcup\limits_{i=1}^{n}{\delta (q={{A}_{i}})}{{:}}\ d(q,{{B}_{i,j}})=1\
\end{equation}
such that $j\ne i $ and $j\in {{\mathcal{I}}_{s}}$. And:
\begin{equation}\label{dist}
\bigcup\limits_{i=1}^{n}{\delta (q={{A}_{i}})}{{:}}\ d(q,f({{B}_{i,j}},{{C}_{i,j,k}},\ldots ,{{D}_{i,j,\ldots }}))=1
\end{equation}
such that $\ldots \ne j\ne i $ and $ j,\ldots \in {{\mathcal{I}}_{s}} $. Whenever a set $\textbf{c}_{i,j}$ of constants can be found such that Remark \ref{constant} or Theorem \ref{permord} hold, the solution is said to be efficient.
\end{remark}
\begin{figure*}[t]
\centering
\includegraphics[height=6.5cm]{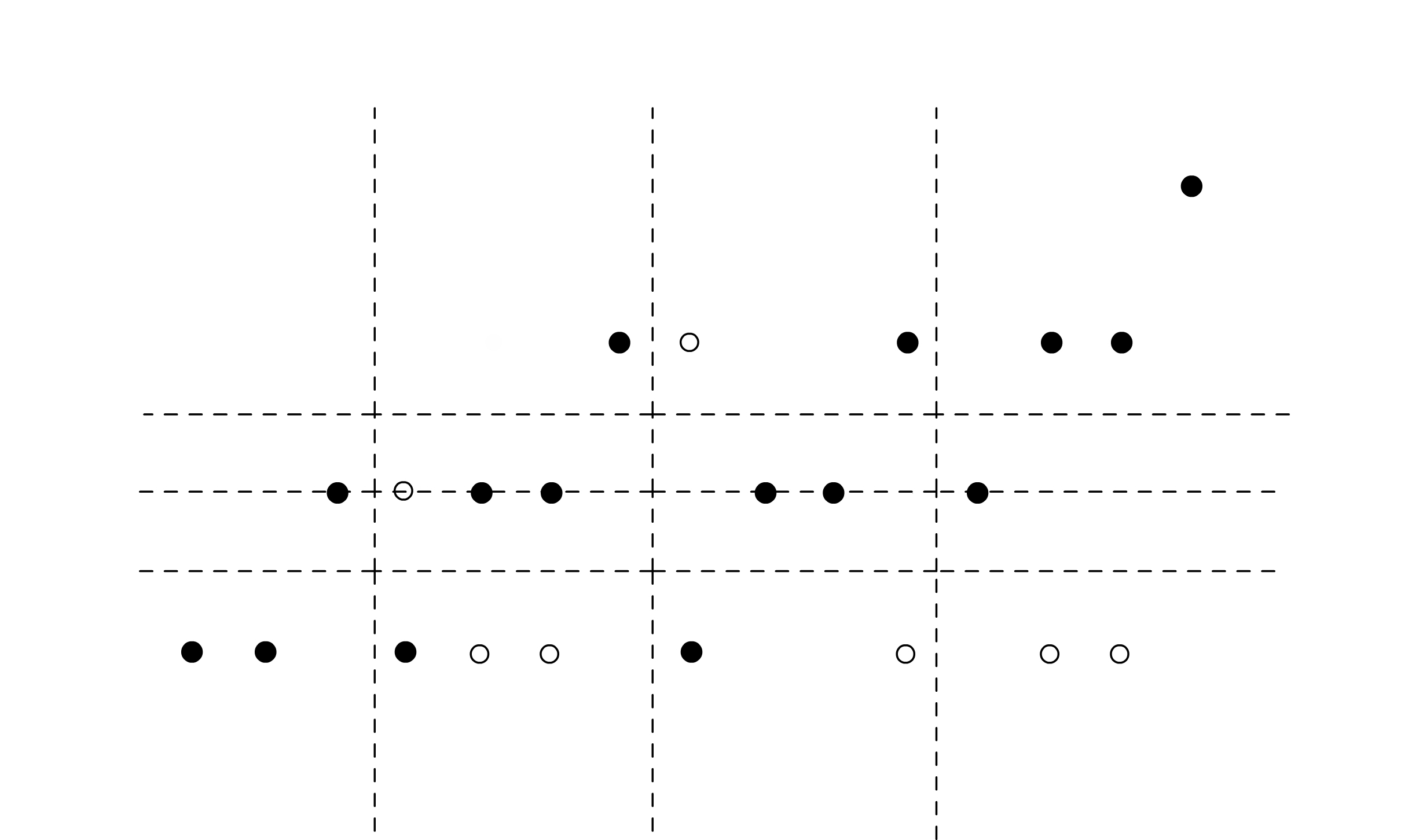}
\caption{Order curve of $L^n_N(\mathcal{I})$.}
\label{orderc}
\end{figure*}
\begin{example}
Starting from the $\lVert\textbf{S}\rVert_0+\lVert\textbf{S}\rVert_1=$ $O(2KM$ $(K+M))$ Jackson-Sheridan-Tseitin representation, the problem can then be stated as such \cite{jackson:cnf,tseitin:cnf}: 

\begin{equation}
\mathbf{S}'=\mathcal{I}\left( \text{L}^2_{K+M}(\mathbf{S})\otimes {{\mathbf{1}}_{{{2}^{2KM}-1}}} \right)
\end{equation}
such that $s_{i,j}\in\{0,1\}$ and $\textbf{s}_i\in\textbf{S}$. Three $l$-codes can be constructed. The first one, on the decimal binary representation of $\textbf{S}$:
\begin{equation}
\text{L}^2_{K+M}(\textbf{S})=\sum_{i=1}^{K+M}\left\lfloor\mathcal{I}_i \log p_i\right\rfloor_{\log}
\end{equation}
A second one on $f(\textbf{S}')$:
\begin{equation}
\text{L}^{2}_{2KM}\circ f\circ\text{L}^2_{K+M}(\textbf{S}_i)=\sum_{i=j}^{2KM}\left\lfloor \mathcal{I}_ic_i\text{L}^2_{K+M}(\textbf{S}_i)\log p_i\right\rfloor_{\log}
\end{equation}
such that:
\begin{equation}\label{function}
c_i\text{L}^2_{K+M}(\textbf{S}_i)\ (\approx,\ge)\ 1
\end{equation}
This satisfies Remark \ref{constant} and a candidate set of constants $c_i$ can then be given as:
\begin{equation}\label{constants}
\left\lfloor\left(\text{L}^2_{K+M}(\textbf{S}_i)\right)^{-1}\left(1+\frac{i}{O((K+M)^{\lfloor \log 10 \rfloor})}\right)\right\rfloor_{\log}
\end{equation}
A third $l$-code on $\textbf{S}'$ is given by:
\begin{equation}
\text{L}^{2KM+1}_{2KM(K+M)}(\textbf{S}')=\sum_{i=1}^{K+M}\left\lfloor k_i\log p_i,\ \ k_i\in[0,k] \right\rfloor_{\log}
\end{equation}
In practice, rational-valued polylogarithmic functions can be used.

The following properties hold on the order curves:
\begin{equation}\label{prop1}
\text{Or}\left(\text{L}^2_{K+M}(\textbf{S})\right) \subset \text{Or}\left(\text{L}^n_N(\mathcal{I})\right) 
\end{equation}
\begin{equation}\label{prop2}
\text{Or}\left(\text{L}^2_{2KM}(f(\textbf{S}')\right)=\text{Or}\left(\text{L}^2_{K+M}(\textbf{S})\right)
\end{equation}
\begin{equation}\label{prop3}
\text{Or}\left(\text{L}^{2KM+1}_{2KM(K+M)}\right) = \textbf{A}/ \textbf{B}
\end{equation}
such that $\textbf{A}=\{0,\ldots,\prod_{i=1}^{2KM} p_i^{k}\}$ and $\textbf{B}$ are the numbers factoring to at least a $p_j$ such that $j\neq i$. The following cardinalities can then be computed:
\begin{equation}\label{prop4}
\begin{split}
\lvert\textbf{A}\rvert=(K+M+1)^{K+M}\\
\lvert\textbf{B}\rvert=n-(K+M+1)^{K+M}
\end{split}
\end{equation}

The distribution of $\text{Or}\left(\text{L}^{2KM+1}_{2KM(K+M)}\right)$ behaves as follows:
\begin{equation}\label{dirac}
\lim_{K+M\rightarrow\infty}\frac{\Delta\text{Or}\left(\text{L}^{2KM+1}_{2KM(K+M)}(\textbf{S}')\right)}{\Delta\text{L}^{2KM+1}_{2KM(K+M)}(\textbf{S}')}=1
\end{equation}
Given Equation \ref{dirac}, as the problem grows larger, the sparsity of the order curve decreases. Furthermore, the distributions of finite differences has a factorial growth complexity, It can then be said that the space complexity of the associated $l$-code is given by: 
\begin{equation}
\left\lvert\text{L}^{2KM+1}_{2KM(K+M)}(\textbf{S}'_i)\right\rvert = O(k\log\lVert\textbf{p}\rVert_\infty)
\end{equation}

Following this, a general form of an associated distance equation is given by:
\begin{equation}\label{solutionJST}
\bigcup_{i=1}^{2KM}\delta({\textbf{q}=A_i}):\ \textbf{q}-B_{i,j}=T
\end{equation}
An actual equation can be given by:
\begin{equation}
\begin{split}
\bigcup_{i\in\mathcal{I}_s}& \delta(q=\sum_{i=1}^{K+M}\left\lfloor a_i\log p_i \right\rfloor_{\log}):\ q+\sum_{\substack{j\neq i\\ j\in\mathcal{I}_s}}\sum_{k=1}^{K+M}\left\lfloor a_{j,k}\log_{p_{j,k}}\right\rfloor_{\log}\\
&\approx\sum_{i=1}^K\left\lfloor \log p_{i} \right\rfloor_{\log}+3\sum_{i=K+1}^{M}\left\lfloor \log p_{i} \right\rfloor_{\log}
\end{split}
\end{equation}
Alternatively, in $l$-code notation:
\begin{equation}
\begin{split}
\bigcup_{i\in\mathcal{I}_s}\delta(q):\ q+\mathcal{I}\text{L}^{2KM+1}_{2KM(K+M)}(\mathcal{I})=T
\end{split}
\end{equation}
Since the zero distance never occurs, Equation \ref{solutionJST} can be termed a distance equation. The approximation factor can be set to equality using Theorem \ref{code}.
\begin{property}
A totient function can be given to count the elements in a prime or log-prime order:
\begin{equation}
\phi(n,\textbf{p})=n\prod_{p_i\in\textbf{p}}\left(1-\frac{1}{p_i}\right)
\end{equation}
such that $gcd(n,p_i)\neq 1$.
\end{property}
In practice, arithmetic progressions or some arbitrary ordered sequence of primes are sufficient, or alternatively, a set of $p$-adic numbers.
\end{example}
This result can be generalized to factorial domains as defined in Definition \ref{fac}.
\begin{property}
The following on ordered primes holds:
\begin{equation}\label{order}
\sum_{i\in\mathcal{I}} \log p_i < \sum_{j\in\mathcal{J}} \log p_{j},\ \ \mathcal{I}\subset\mathcal{J}
\end{equation}
\end{property}
\begin{remark}\label{sigmoid}
Let the following sequence of iterates be defined on the order given by Equation \ref{prop3}:
\begin{equation}
f_i=x,\ \ x\in\text{Or}\left(\text{L}^{2KM+1}_{2KM(K+M)}\right)
\end{equation}
It can be shown that $\log_i \log_i f_i $ is approximately an inverse sigmoidal function. 
\end{remark}
\begin{theorem}
$\text{Or}(\text{L}_{N}^{n}(\mathcal{I}))$ can be computed using $O(\log N)$ recursions.
\end{theorem}
\begin{theorem}
On a logarithmic factorial prime domain, the space complexity of a finite discrete space $\textbf{X}$ is upper bounded by:
\begin{equation}
O(\log\lVert\textbf{X}\rVert_\infty)
\end{equation}
\end{theorem}
\begin{theorem}\label{code}
A code of cardinality $k$ in some distance space is reversible if it provides a $k$-partition.
\end{theorem}
The order on the $l$-codes of some indicator function can be computed without the knowledge of the associated $l$-codes. The indicator function defines a space of complexity $O(2^n)$ or $O(n^n)$. The complete search can then be done in approximately $O(n\log 2)$ and $O(n\log n)$ respectively. A relaxation of the problem only requires the knowledge of the number of elements before and after a given element in the $l$-code. This can be done in $O(1)$. The upper-bound provided on the finite difference of the class of integral prime $l$-codes then provides a bound on the space complexity of the coding.

\subsection{Further Remarks on g-Codes and l-Codes}
\subsubsection{Iterated Encodings}
\begin{definition} 
Let $n$ be the length of a code, $N$, the number of codes in a (g,l)-code, $i$, the number of iterations and $j$ the number of (g,l)-codes in a given iteration, an iterated encoding is then recursively defined as follows:
\begin{equation}                                                   {}_{j}^{i}\text{G}_{N}^{n}={}_{{{j}_{1}}}^{{{i}_{1}}}\text{G}_{{{N}_{1}}}^{{{n}_{1}}}\circ \bigcup{{}_{{{j}_{2}}}^{{{i}_{2}}}\text{G}_{{{N}_{2}}}^{{{n}_{2}}}\circ \ldots \circ \bigcup{{}_{{{j}_{k}}}^{{{i}_{k}}}\text{G}_{{{N}_{k}}}^{{{n}_{k}}}}}	
\end{equation}
And similarly:
\begin{equation}	                                                   {}_{j}^{i}\text{L}_{N}^{n}={}_{{{j}_{1}}}^{{{i}_{1}}}\text{L}_{{{N}_{1}}}^{{{n}_{1}}}\circ \bigcup{{}_{{{j}_{2}}}^{{{i}_{2}}}\text{L}_{{{N}_{2}}}^{{{n}_{2}}}\circ \ldots \circ \bigcup{{}_{{{j}_{k}}}^{{{i}_{k}}}\text{L}_{{{N}_{k}}}^{{{n}_{k}}}}}
\end{equation}
\end{definition}	
Iterated encodings preserve the complexity of a single iteration whenever the number of iterations is $O(1)$. This is denoted:
\begin{equation}
ECDC\left( {}_{j}^{i}\text{G}_{N}^{n} \right)=O\left( ECDC\left( \text{G}_{O(\lVert\textbf{N}\rVert_\infty)}^{O(\lVert \textbf{n} \rVert_\infty)} \right) \right)	
\end{equation}
\begin{figure*}
\centering
\includegraphics[height=19cm]{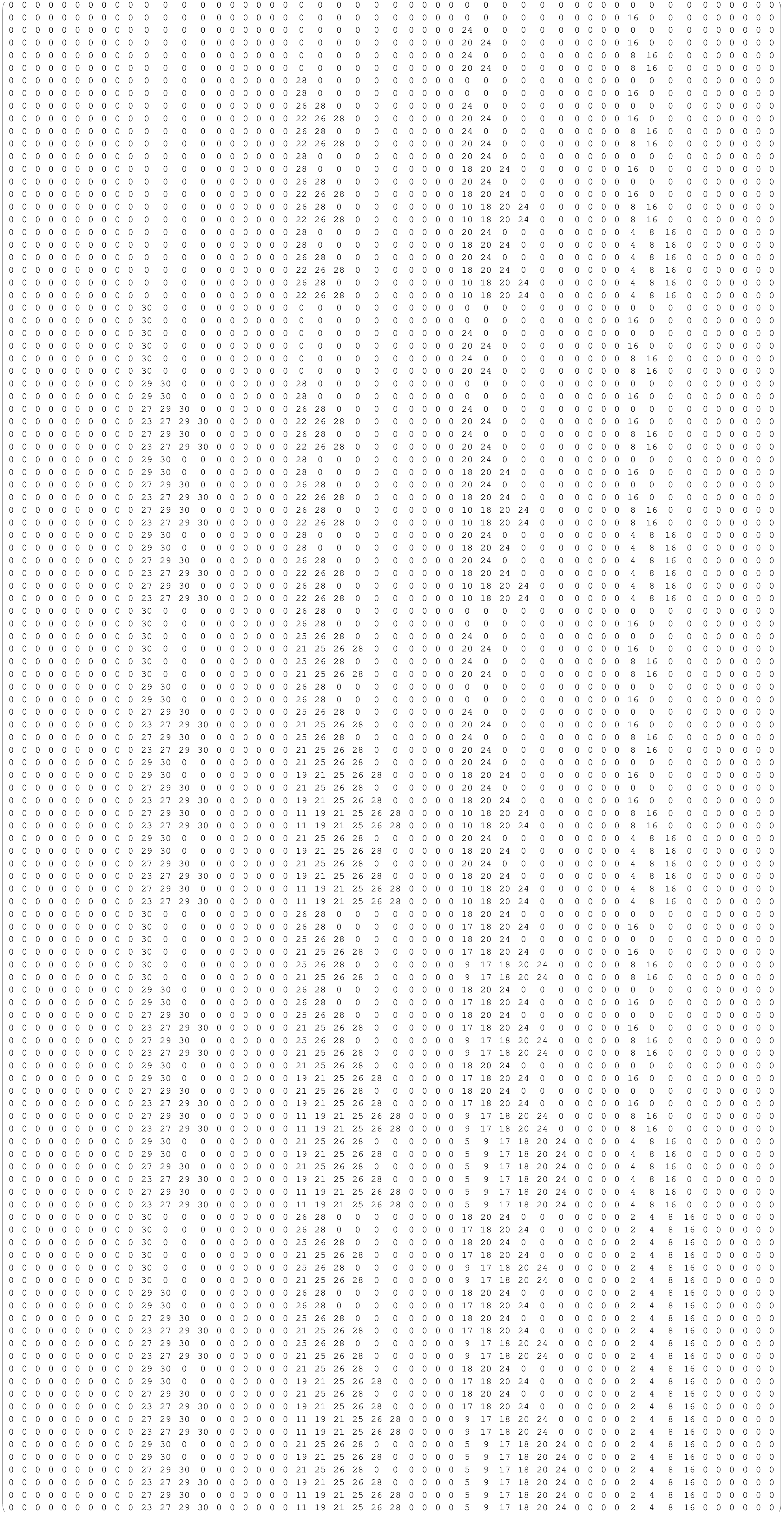}
\caption{Matrix values of $f(\pi(1,2,3,4,5))$.}
\label{Numbers}
\end{figure*}

\subsubsection{Limiting Encodings}
An infinite sequence is incompressible using $g$-codes, which follows the classic argument in algorithmic complexity that shows that the problem of defining such a code is in general uncomputable. The logarithmic upper-bound of $l$-codes follows the classic information theoretical argument. 
	
The basic results of algorithmic complexity can be restated as such for codes and sequences.
\begin{theorem}
An infinite code may or may not have a finite algorithmic encoding.
\end{theorem}
\begin{theorem} 
The number of infinite codes that have a finite algorithmic encoding is finite.
\end{theorem}
\begin{theorem} 
The problem of determining the existence of a finite algorithmic encoding for an infinite code is decidable.
\end{theorem}
In the context of finite discrete spaces, the main result is as follows.
\begin{theorem}
Finite codes have an algorithmic encoding upper-bounded by $\log N$ where $N$ is given by the ${{\ell }_{\infty }}$ norm.
\end{theorem}

\subsubsection{l-Functions}
An l-function is defined as a mapping from an arbitrary encoding to a set of $l$-codes. Formally:
\begin{definition}
Let l be an l-function. The following then holds:
\begin{equation}
\l(\mathbf{X})=\left\{ {}_{{{j}_{1}}}^{{{i}_{1}}}\text{L}_{{{N}_{1}}}^{{{n}_{1}}},\ldots ,{}_{{{j}_{k}}}^{{{i}_{k}}}\text{L}_{{{N}_{k}}}^{{{n}_{k}}} \right\}
\end{equation}
\end{definition}
Reversibility and uniqueness properties can be discussed in the context of l-functions.
\begin{definition}
An l-function is said to be reversible if and only if the following holds:
\begin{equation}
f(\mathbf{X})=\bigcup\limits_{k}{{{l}_{k}}({{\mathbf{X}}_{k}})}
\end{equation}	
for some bijective function $f$ defining a unique $l$-code.
\end{definition}

\subsubsection{Complete l-Functions}
Complete encodings define as usual the set of encodings that are Turing complete.
\begin{theorem}\label{tc}
The following sequence of iterates illustrates a Turing machine:
\begin{equation}
{{\mathbf{x}}_{i}}={{\mathbf{x}}_{i-1}}+{{\mathbf{r}}_{i}}{{,}^{{}}}^{{}}{}^{{}}i=1,\ldots 
\end{equation}	
with the following:
\begin{equation}
\bigcup\limits_{k}{{{l}_{k}}({{\mathbf{x}}_{k}})=\bigcup\limits_{k}{{}_{{{j}_{k}}}^{{{i}_{k}}}\text{L}_{{{N}_{k}}}^{{{n}_{k}}}}
(\mathbf{x})}
\end{equation}
\end{theorem}
	The $l$-code representation of a universal encoding is denominated a complete l-function. Theorem \ref{tc} provides a class of simple l-functions built on the permutation space of an algorithmic encoding.

\subsubsection{Application of l-Function on (1,0)}
Let some function $f$ define a sorting algorithm known as bead sort \cite{calude:bead} on $N$ random integers $\mathbf{B}$ in their matrix rank order representation such that:
\begin{equation}
\bigcup\limits_{k,j=1,N}^{M,1}f_{i}(b_{j,k},b_{j+1,k})=\left\{ 
\begin{array}{rl}
   \begin{matrix}
   \ddots  & \vdots  & \iddots \\
\ldots  &  {{b}_{j+1,k}}  & \ldots \\
 \ldots  &  {{b}_{j,k}}  & \ldots \\
 \iddots  & \vdots  & \ddots 
\end{matrix}, & (b_{j,k},b_{j+1,k})=(1,0)   \\
   \begin{matrix}
   \ddots  & \vdots  & \iddots \\
\ldots  &  {{b}_{j,k}}  & \ldots \\
 \ldots  &  {{b}_{j,k+1}}  & \ldots \\
 \iddots  & \vdots  & \ddots 
\end{matrix}, & otherwise   
\end{array}\right.
\end{equation}
and:
\begin{equation}
\mathbf{B}=\left[ 
\renewcommand{\arraystretch}{1.5}
\begin{matrix}
   {{b}_{1,1}} & \ldots  & {{b}_{1,M}}  \\
   {{b}_{1,2}} & \ldots  & {{b}_{2,M}}  \\
   \vdots  & \ddots  & \vdots   \\
   {{b}_{N,1}} & \ldots  & {{b}_{N,M}}  \\
\end{matrix} \right]
\end{equation}
with ${{b}_{i,j}}\in \{0,1\}$.

Then the following representation can be used:
\begin{equation}\label{lint}
f_i=\sum_{i=1}^{N!} \frac{a_i}{b_i}^i
\end{equation}
Equation \ref{lint} can be denoted as the $\mathcal{L}$-interpolation of $f$ and defines the upper-bound of   on the following complexity space:
\begin{equation}
\renewcommand{\arraystretch}{1.5}
\begin{array}{c|c}
   O(f) & O(|f|)  \\ \hline
   O(f) & O(N!)  \\
   \vdots  & \vdots   \\
   O(N) & O(1)  \\
\end{array}
\end{equation} 

This is shown in Figure \ref{Numbers}.


\section{Acknowledgements}

\emph{I would like to thank all the people who have given me their support during the writing of this paper done during my Ph.D. work in Computer Science.}

\end{document}